%% file: noncomm-bb.tex
\documentclass[submission,copyright,creativecommons]{eptcs}
\providecommand{\event}{QPL 2014}

\title{Tensors, !-graphs, and non-commutative quantum structures}
\author{Aleks Kissinger
\institute{University of Oxford}
\email{aleks.kissinger@cs.ox.ac.uk}
\and
David Quick
\institute{University of Oxford}
\email{david.quick@cs.ox.ac.uk}}
\def\titlerunning{Tensors and !-graphs}
\def\authorrunning{A. Kissinger \& D. Quick}

\usepackage{bussproofs}
\input{preamble.tex}

\renewcommand{\showproofs}{0}

\begin{document}

\maketitle

\begin{abstract}
	Categorical quantum mechanics (CQM) and the theory of quantum groups rely heavily on the use of structures that have both an algebraic and co-algebraic component, making them well-suited for manipulation using diagrammatic techniques. Diagrams allow us to easily form complex compositions of (co)algebraic structures, and prove their equality via graph rewriting. One of the biggest challenges in going beyond simple rewriting-based proofs is designing a graphical language that is expressive enough to prove interesting properties (e.g. normal form results) about not just single diagrams, but entire families of diagrams. One candidate is the language of \textit{!-graphs}, which consist of graphs with certain subgraphs marked with boxes (called !-boxes) that can be repeated any number of times. New !-graph equations can then be proved using a powerful technique called \textit{!-box induction}. However, previously this technique only applied to commutative (or cocommutative) algebraic structures, severely limiting its applications in some parts of CQM and (especially) quantum groups. In this paper, we fix this shortcoming by offering a new semantics for \textit{non-commutative} !-graphs using an enriched version of Penrose's abstract tensor notation.
\end{abstract}

\section{Introduction}\label{sec:intro}

\textit{Diagrammatic theories} give us a way to study a wide variety of algebraic and coalgebraic structures in monoidal categories. They consist of two parts: a \textit{signature} $\Sigma$ and a set of \textit{diagram equations} $E$. The signature consists of a set of objects $\{ A, B, \ldots \}$ along with a set of generating morphisms with input and output arities formed from combining objects with $\otimes$ and $I$. For example, the signature of a Frobenius algebra consists of four morphisms: $(\mu : A \otimes A \to A,\ \eta : I \to A,\ \delta : A \to A \otimes A,\ \epsilon : A \to I)$, or, written diagrammatically:
\[ \Sigma \ =\  %
\beginpgfgraphicnamed{frob_sig}
\InputIfFileExists{frob_sig.tikz}{}{\input{./figures/frob_sig.tikz}}
\endpgfgraphicnamed \]
Then, $E$ is a set of equations between morphisms built from these generators, which we can picture as equations between string diagrams. For example, the theory of commutative Frobenius algebras contains the (co)associativity, (co)unit, (co)commutativity and Frobenius equations: \\
\begin{equation}\label{eqn:frob-laws}
\beginpgfgraphicnamed{comm_frob_eqns}
\InputIfFileExists{comm_frob_eqns.tikz}{}{\input{./figures/comm_frob_eqns.tikz}}
\endpgfgraphicnamed
\end{equation}
A \textit{model} of $(\Sigma, E)$ in a (symmetric, traced, or compact closed) monoidal category $\mathcal C$ assigns a morphism to each generator in $\Sigma$ such that all equations in $E$ hold.

\begin{remark}
	Many familiar algebraic constructions arise as special cases of this setup. For instance, any linear `term-like' algebraic theory (i.e. where free variables occur precisely once on the LHS and RHS of every equation) can be presented this way. Also, if we restrict to equations in $E$ that are directed acyclic, we obtain presentations of PROPs (or coloured PROPs in the multi-sorted case). In that case, models of $(\Sigma, E)$ in $\mathcal C$ are in 1-to-1 correspondence with strong monoidal functors from the presented PROP into $\mathcal C$.
\end{remark}

This style of algebraic theory works well when generators have fixed, finite arity. However, it is often possible to find a much more elegant presentation of a theory if we allow the arity of our generators to vary. For instance, commutative Frobenius algebras can be alternatively presented using a single variable-arity generator sometimes called a `spider', along with just two equations.
\[
\Sigma \ =\ \left\{ \ %
\beginpgfgraphicnamed{spider-gen}
\InputIfFileExists{spider-gen.tikz}{}{\input{./figures/spider-gen.tikz}}
\endpgfgraphicnamed \ \right\} \qquad
E \ =\ \left\{ \ %
\beginpgfgraphicnamed{spider_merge}
\InputIfFileExists{spider_merge.tikz}{}{\input{./figures/spider_merge.tikz}}
\endpgfgraphicnamed, \ \ \ \
\beginpgfgraphicnamed{spider_elim}
\InputIfFileExists{spider_elim.tikz}{}{\input{./figures/spider_elim.tikz}}
\endpgfgraphicnamed \ \right\}
\]
A model of such a theory is no longer just a finite set of morphisms, but rather, a set of \textit{families} of morphisms $f_{j,k} : A^{\otimes j} \to A^{\otimes k}$, indexed by input/output arities, such that the equations in $E$ hold for all possible arities.

Comparing this to the equations at the beginning of this section, we seem to have lost some formality. That is, the `concrete' diagrammatic identities above can be formalised in such a way that proofs can be performed (and even machine-checked) via a suitable notion of diagram rewriting, as formalised in~\cite{DixonKissinger2010}. One might be tempted to think that this level of rigour is lost when we describe equations in a mathematical meta-language, making use of ellipses, for example, to represent repetition. However, in~\cite{DixonDuncan2009}, the authors introduced \textit{!-boxes} (pronounced `bang-boxes') as a method for reasoning about graphs with repeated structure. As !-box rules, the previously informal rules can be formalised as:
\[
\Sigma \ =\ \left\{ \ %
\beginpgfgraphicnamed{spider_bb}
\InputIfFileExists{spider_bb.tikz}{}{\input{./figures/spider_bb.tikz}}
\endpgfgraphicnamed \ \right\} \qquad
E \ =\ \left\{ \ %
\beginpgfgraphicnamed{spider_merge_bb}
\InputIfFileExists{spider_merge_bb.tikz}{}{\input{./figures/spider_merge_bb.tikz}}
\endpgfgraphicnamed, \ \
\beginpgfgraphicnamed{spider_elim}
\InputIfFileExists{spider_elim.tikz}{}{\input{./figures/spider_elim.tikz}}
\endpgfgraphicnamed \ \right\}
\]
Intuitively, marking a subgraph with a !-box means that subgraph (along with edges in/out of it) can be repeated any number of times to obtain an \textit{instance} of the graph. Thus we interpret a graph with !-boxes as a set of all its instances.
\[ \left\llbracket %
\beginpgfgraphicnamed{spider_bb}
\InputIfFileExists{spider_bb.tikz}{}{\input{./figures/spider_bb.tikz}}
\endpgfgraphicnamed \quad \right\rrbracket
   \ :=\  \left\{ \ \ %
\beginpgfgraphicnamed{spider_inst1}
\InputIfFileExists{spider_inst1.tikz}{}{\input{./figures/spider_inst1.tikz}}
\endpgfgraphicnamed \ \ \right\} \]
Similarly, for rules with !-boxes, matched pairs of !-boxes can be repeated in the LHS and RHS to obtain instances of that rule. Thus, for our example of the commutative Frobenius algebra, we have reduced our theory of 7 equations to just 2.

!-boxes were given a formal semantics in~\cite{PatternGraph}, making use of adhesive categories~\cite{LackAdh2005}. They also come with a simple and powerful induction principle introduced by one of the authors in~\cite{KissingerThesis} and proven correct in~\cite{MerryThesis}. But there's a catch: note how we were careful to say that \textit{commutative} Frobenius algebras have an elegant presentation as above. A major drawback of the existing !-box notation is that it is only unambiguous if all of the nodes in the diagram are invariant under permuting inputs/outputs. This is severely limiting in two ways. The first and most obvious limitation is that we are forced to consider only commutative algebraic structures. The second, more subtle limitation is that we have no freedom to \textit{definitionally} extend our theory, i.e. introduce new nodes defined as diagrams of other nodes, without making implicit assumptions about those diagrams (namely, that they are symmetric on inputs/outputs).

In order to overcome these shortcomings, we extend the !-graph notation with some extra information about how newly-created edges should be ordered when a !-box is expanded. This turns out to be fairly straightforward as soon as one shifts from a graph-based semantics for diagrams, as employed in~\cite{DixonKissinger2010}, to a \textit{tensor-based} semantics, where morphisms in the free compact closed category are represented using a version of Penrose's abstract tensor notation~\cite{Penrose1971}. This approach, recently formalised in~\cite{KissingerATS}, has the property that non-commutativity comes `for free', where the edges connected to a single element are represented as a list of edge names. Contrast this with the graph-based semantics for string diagrams or Joyal and Street's geometric construction~\cite{JS}, where one needs to add some extra structure (e.g. a total ordering or typing on adjacent edges) to break symmetries.



So, without further ado, we introduce tensor expressions for compact closed categories and extend them to accommodate !-boxes.


\section{Tensors}\label{sec:tensors}

Assume we are working in a compact closed category $\mathcal C$ freely generated by a set of objects $X, Y, Z, \ldots$ and morphisms of the form $\phi : I \to X_1 \otimes \ldots \otimes X_n$, i.e. morphisms with only non-trivial outputs. Since $\mathcal C$ is compact closed, this yields no loss of generality, since we represent an input of type $A$ as an output of type $A^*$. For simplicity, we'll assume every `input' is of fixed type $X^*$ and every `output' is of type $X$.

Since we want to distinguish inputs/outputs we label them using lower case letters. They will have a hat to illustrate being an `output': $\{\eout{a},\eout{b},\ldots\}$, or a check to illustrate being an `input': $\{\ein{a},\ein{b},\ldots\}$. Translating a morphism $\phi$ into tensor notation yields:
\[
\phi : I \to X \otimes X \otimes X^* \otimes X^* \otimes X^* \qquad \Longrightarrow \qquad
\term{\phi.+a+b-c-d-e.}
\]

\begin{wrapfigure}[4]{r}{0.4\textwidth}
\vspace{-16pt}
$%
\beginpgfgraphicnamed{phi-point}
\InputIfFileExists{phi-point.tikz}{}{\input{./figures/phi-point.tikz}}
\endpgfgraphicnamed \quad \Longrightarrow \quad
\beginpgfgraphicnamed{compl-fixed-arity}
\InputIfFileExists{compl-fixed-arity.tikz}{}{\input{./figures/compl-fixed-arity.tikz}}
\endpgfgraphicnamed$
\end{wrapfigure}
We introduce a special graphical notation for morphisms with only outputs. We write them as circles with a tick, taking the convention that inputs/outputs are ordered clockwise from the tick.


Writing two tensors side-by-side yields a new tensor formed by taking the monoidal product and `contracting' any repeated names using the compact structure on $X$.
\vspace{-10pt}
\begin{equation}\label{eq:contraction}
  \term{\psi.+f\bR-a-b\e.\phi.\bR+a+b\e-c-d-e.}
\quad := \quad
\beginpgfgraphicnamed{psi-phi-contract}
\InputIfFileExists{psi-phi-contract.tikz}{}{\input{./figures/psi-phi-contract.tikz}}
\endpgfgraphicnamed
\quad \Longrightarrow \quad
\beginpgfgraphicnamed{compose-graphs}
\InputIfFileExists{compose-graphs.tikz}{}{\input{./figures/compose-graphs.tikz}}
\endpgfgraphicnamed
\end{equation}

We say repeated edge names (e.g. $a$ and $b$ above) are \textit{bound} in a tensor expression, and all other edge names are \textit{free}. In the graph we have labelled the bound edges, though this is purely for demonstrating which edges are bound. The names of bound edges can be changed at will, provided they are replaced with new, fresh names. Hence $\term{\psi.+f\bR-a-b\e.\phi.\bR+a+b\e-c-d-e.}$ and $\term{\psi.+f\bR-x-y\e.\phi.\bR+x+y\e-c-d-e.}$ represent the same graph. As a result, we typically will not write down bound names in the graphical notation.

\begin{definition}
	The set of \textit{tensor expressions} for a signature $\mathcal S$ consists of (i) the trivial tensor $1$, (ii) the identity tensor $\term{1.+a-b.}$, (iii) atomic tensors $\term{\psi.+a-b\ldots.}$ with the appropriate names for each $\psi \in \mathcal S$, (iv) $GH$ for $G, H$ tensor expressions, and (v) $G'$ obtained by changing some of the names of a tensor expression $G$---subject to the condition that $\eout{a}$ and $\ein{a}$ occur at most once for each name $a$.
\end{definition}

\begin{definition}\label{def:tensor-equiv}
	Two tensor expressions $G, G'$ are equivalent, written $G \equiv G'$ if $G$ can be made into $G'$ by replacing bound names or by applying one or more of the following identities:
	\[
	    (GH)K \equiv G(HK) \qquad
	    GH \equiv HG \qquad
	    G1 \equiv G
  \]
  \[
      G\term{1.+b-a.} \equiv G[\ein b \mapsto \ein a] \qquad
      H\term{1.+a-b.} \equiv H[\eout b \mapsto \eout a]
	\]
  Assume for the last two identities that $\ein b$ and $\eout b$ are free in $G$ and $H$, respectively. An $\equiv$-equivalence class of tensor expressions is called a \textit{tensor}.
\end{definition}

Note that we use $\equiv$ for syntactic equivalence of tensor expressions (and later !-tensor expressions). We reserve the normal equals sign for equality by the rules of a given theory. As such, we always assume $(G \equiv H) \implies (G = H)$, but not the converse.

Tensors are related to morphisms in the free compact closed category as follows. Suppose we fix a set of \textit{canonical names} $\{ \ein{x}_1, \ein{x}_2, \ldots \}$ and $\{ \eout{x}_1, \eout{x}_2, \ldots \}$. A tensor $G$ is said to be \textit{canonically named} if for some $N$ it has as a free name precisely one of $\ein{x}_i$ or $\eout{x}_i$ for $1 \leq i \leq N$.

\begin{theorem}
	Canonically-named tensors are in 1-to-1 correspondence to morphisms in the free compact closed category generated by a signature $\mathcal S$.
\end{theorem}

\begin{proof}
	First note that adding `hats' and `checks' to edge names is essentially applying the Int construction (c.f.~\cite{JSV}) to free traced symmetric monoidal category, in the tensorial presentation given in~\cite{KissingerATS}. The free compact closure of the free traced monoidal category then satisfies the appropriate universal property to make it the free compact closed category.
\end{proof}

To summarise, we can interpret a tensor in a compact closed category as follows. First, we swap its free names for `canonical names' (or otherwise order the outputs somehow), then interpret each atomic expression as a morphism (or one of a family of morphisms, parametrised by its arity). Finally, we construct the composed morphism by composing each of the components and contracting repeated edge names, as in~\eqref{eq:contraction}.

Alternatively, one can study models in an existing abstract tensor system (in the sense of Penrose), in which case interpretation is trivial. These two points of view (categorical vs. ATS) are roughly equivalent, as was shown in~\cite{KissingerATS}.

\section{Adding !-boxes to tensor expressions}

We now extend the existing tensor notation with !-boxes. Graphically !-boxes are blue boxes surrounding a subgraph, labelled with a name ($A,B,\ldots$). We can denote this with square brackets around a subterm in a tensor expression, labelled with a superscript. Intuitively a !-box represents a portion of the graph that can be copied multiple times. For this to be well-defined in the non-commutative case we need to clarify where each new copy of the subgraph gets attached to surrounding nodes.

This is done by assigning an expansion direction (clockwise vs anticlockwise) to any group of edges from a node to a !-box. We draw these as arrows over edge groups in our !-graphs and for our tensors we denote clockwise edge groups as $\eclock{\dots}{A}$ and anticlockwise edge groups as $\eanti{\ldots}{A}$. For example:
\[
\term{\phi.<+a]B.[\psi.-a.]B} \ :=\  %
\beginpgfgraphicnamed{bb-example}
\InputIfFileExists{bb-example.tikz}{}{\input{./figures/bb-example.tikz}}
\endpgfgraphicnamed \qquad\textrm{vs.}\qquad
\term{\phi.[+a>B.[\psi.-a.]B} \ :=\  %
\beginpgfgraphicnamed{bb-example-cw}
\InputIfFileExists{bb-example-cw.tikz}{}{\input{./figures/bb-example-cw.tikz}}
\endpgfgraphicnamed
\]
In the next section, we will see how the arrows clarify not only which direction edges should expand, but also whether they should expand in groups or individually. For example, the following notation gives anti-clockwise expansion of $\term{+a-b}$ as a group, clockwise expansion of $\term{+a-b}$ as a group, and clockwise expansion of $\term{+a}$ and $\term{-b}$ as individual edges, respectively:
\[
\begin{matrix}
  \term{\psi.+a'-b'<+a-b]A.[]A} & &
  \term{\psi.[+a-b>A+a'-b'.[]A} & &
  \term{\psi.[+a>A +a'[-b>A -b'.[]A} \\
  & & & & & \\
\beginpgfgraphicnamed{expand-order1}
\InputIfFileExists{expand-order1.tikz}{}{\input{./figures/expand-order1.tikz}}
\endpgfgraphicnamed & 
  \quad\textrm{vs.}\quad &
\beginpgfgraphicnamed{expand-order2}
\InputIfFileExists{expand-order2.tikz}{}{\input{./figures/expand-order2.tikz}}
\endpgfgraphicnamed &
  \quad\textrm{vs.}\quad &
\beginpgfgraphicnamed{expand-order3}
\InputIfFileExists{expand-order3.tikz}{}{\input{./figures/expand-order3.tikz}}
\endpgfgraphicnamed
\end{matrix}
\]

It is also possible for !-boxes to be nested inside other !-boxes. This means expansion of the parent box makes a new copy of the child with a new !-box name. Edge groups can correspondingly
\begin{wrapfigure}[7]{l}{0.15\textwidth}
	\vspace{-10pt}
	\ \ %
\beginpgfgraphicnamed{nested-ex}
\InputIfFileExists{nested-ex.tikz}{}{\input{./figures/nested-ex.tikz}}
\endpgfgraphicnamed
\end{wrapfigure}
\noindent be nested if the edges enter more than one box. In the diagram to the left we have the !-graph with !-tensor expression: $\term{\phi.+a[<-b]B>A.[[\phi.+b-c.]B]A}$. We have labelled which arrow corresponds to which !-box. This is not necessary if we adopt the convention that a parent box's arrow must be drawn closer to the node than it's child box's arrow. Note that the labels inside nodes are to assign a type to the node as apposed to naming the node. This means since we often have multiple nodes with the same type, we will have nodes with the same label.

We can now imagine more general generators allowing arbitrary arrangements of input and output edges. Any such node, say of type $\phi$, then needs to be assigned a morphism in our category for each possible arrangement of edges. We represent an arrangment as a word over $\{\wedge,\vee\}$ where $\wedge$ represents outputs and $\vee$ represents inputs. For example the node $%
\beginpgfgraphicnamed{outinoutout}
\InputIfFileExists{outinoutout.tikz}{}{\input{./figures/outinoutout.tikz}}
\endpgfgraphicnamed$ has edge arrangment $\wedge\vee\wedge\:\wedge$ and needs to be assigned a morphism $f:I\to X \otimes X^* \otimes X \otimes X $. Hence we need $\phi:\{\wedge,\vee\}^*\to \text{Mor}(\mathcal{C})$ to model the node type $\phi$.

!-tensors replace lists of edges on individual morphisms with \textit{edgeterms} of which we now give a recursive definition.

\begin{definition}\label{def:edgeterm-equiv}
Fix a disjoint, infinite sets $\edgenames$ and $\boxnames$ of edge names and !-box names, respectively. We denote the set of \textit{directed edges} as $\diredgenames := \{\ein{a},\eout{a} : a\in\edgenames\}$. The set of \textit{edgeterms} $\edgeterms$ is defined recursively as follows:
	\begin{align*}
		\bullet\; &\epsilon \in \edgeterms && \text{(i.e empty)} \\
		\bullet\; & \ein a, \eout a \in \edgeterms && a \in \edgenames \\
		\bullet\; &\term{[e>A},\term{<e]A} \in \edgeterms && e\in\edgeterms,\; A\in\boxnames \\
		\bullet\; & e f \in \edgeterms && e,f\in\edgeterms
	\end{align*}
	Two edgeterms are equivalent if one can be transformed into the other by:
	\[
	\epsilon e \equiv e \equiv e \epsilon \qquad
	e(fg) \equiv (ef)g \qquad
	\term{[\epsilon>A} \equiv \epsilon \equiv \term{<\epsilon]A}
	\]
\end{definition}

Since the well-formedness conditions for !-tensor expressions are a bit more complicated than for tensor expressions, we first define the set of all !-pretensor expressions, including those that may be ill-formed.

\begin{definition}
The set of all !-pretensor expressions $\graphterms'$ for a signature $\Sigma$ is defined recursively as:
  \begin{align*}
    \bullet\; &1, \term{1.+a-b.} \in \graphterms' && a,b\in\edgenames\\
    \bullet\; &\term{\phi.e.} \in \graphterms' && e\in\edgeterms, \phi\in\Sigma \\
    \bullet\; &\term{[G]A} \in \graphterms' && G\in\graphterms', \; A\in\boxnames \\
    \bullet\; & G H \in \graphterms' && G,H\in\graphterms'
  \end{align*}
\end{definition}

We introduce the notion of a \textit{context}, which lists the !-boxes in which a certain edge name occurs, from the inside-out. These come in two flavours, \textit{edge contexts} and \textit{node contexts}.

\begin{definition}
  Given a directed edge $a\in\diredgenames$ in a !-tensor $G$ nested as $\term{[[\phi.\ldots<<a>{E_1}\ldots>{E_n}\ldots.]{N_1}\ldots]{N_m}}$.

  We define the \textit{edge context}, \textit{node context}, and \textit{context} of $a$ respectively as:
    \begin{align*}
      \ectx_G(a)&:=[E_1,\ldots,E_n] && \text{(edge context)} \\
      \nctx_G(a)&:=[N_1,\ldots,N_m] && \text{(node context)} \\
      \ctx_G(a)&:=\ectx_G(a).\nctx_G(a) && \text{(context)}
    \end{align*}
    That is, $\ectx_G(a)$ lists the !-boxes containing $a$ that occur as part of $a$'s edgeterm, and $\nctx_G(a)$ lists the rest.
\end{definition}

Finally, a !-tensor expression is a !-pretensor expression where !-box/edge names must be suitably unique and occur in compatible contexts.

\begin{definition}
  A !-tensor expression is a !-pretensor expression satisfying the following conditions:
  \begin{itemize}
    \item[F1.] $\ein a$ and $\eout a$ occur at most once for each edge name $a$
    \item[F2.] $\term{[\ldots]A}$ must occur at most once for each !-box name $A$
    \item[C1.] $\ectx_G(a)\cap \nctx_G(a)=\varnothing$ \: for all edges $a \in \edgenames$ in $G$
    \item[C2.] If $\ectx_G(a)=[B_1,\ldots,B_n]$ \; then all $B_i\in\Boxes(G)\:$ and $\; B_1\prec_G B_2\prec_G \ldots \prec_G B_n$
    \item[C3.] For all bound pairs $\ein a, \eout a$ of edge names in $G$, there exist lists $es, bs$ of !-box names such that:
    \[ es.\nctx_G(\ein{a}) = \ectx_G(\eout{a}).bs
       \quad \textrm{and} \quad
       es.\nctx_G(\eout{a}) = \ectx_G(\ein{a}).bs \]
  \end{itemize}
  where $A \prec_G B$ means that the !-box $A$ is nested inside $B$ in $G$ (without other boxes nested between). We write $\graphterms$ for the set of all !-tensor expressions for a signature $\Sigma$.
\end{definition}

The freshness conditions F1 and F2 ensure that we have not used the same name for more than one edge/box. If a node is in !-box $B$ then any edges attached to it are already in $B$ so it wouldn't make sense to have $B$ in both the $\ectx(a)$ and $\nctx(a)$ for $a\in\diredgenames$, this is enforced by C1. C2 ensures that edge contexts are compatible with the !-boxes in the rest of the !-tensor. For example $\term{\phi.[[-a>A>B.}$ requires $A$ to be nested in $B$ so does not result in a valid !-tensor when composed with e.g. $\term{[[\psi.-b.]B \xi.<+b]B.]A}$. C3 ensures that edges into !-boxes from the outside are decorated correctly by their edge terms. For instance, this is allowed: $\term{\psi.[+a>A.[\phi.-a.]A}$ but this is not: $\term{\psi.+a.[\phi.-a.]A}$. The freedom to pick $bs, es$ allows bound pairs of edges to share some common context, e.g.: $\term{[\psi.+a.\phi.-a.]A}$ (both nodes are inside $A$) or $\term{\psi.[+a>A.\phi.<-a]A.[]A}$ (only the edge is inside $A$). In the second example, $A$ occurs in an edge term, so C2 requires the presence of $\term{[\ldots]A}$ somewhere in the !-tensor, hence we append the `empty' !-box $\term{[]A}$.

In this paper when we write a composition $GH$, unless stated otherwise, we will assume this forms a well defined !-tensor.

Naturally, we say two !-tensor expressions are equivalent, written $G \equiv H$, if one can be obtained from the other by using the usual tensor equivalences from Definition~\ref{def:tensor-equiv} or by using the edgeterm equivalences from Definition~\ref{def:edgeterm-equiv}.

We call the graphical notation for !-tensors the \textit{non-commutative !-graph notation}, or simply (non-commutative) !-graphs.

\begin{theorem}\label{thm:bang-graph-notation}
  Any !-tensor can be represented unambiguously using non-commutative !-graph notation.
\end{theorem}

\begin{proof}
  We show this by providing a general procedure for interpreting a !-graph as a !-tensor expression, and vice-versa. For the sake of clarity, we demonstrate each step on a worked example. Given a non-commutative !-graph, we wish to obtain a unique equivalence class of !-tensor expressions under $\equiv$. Begin by choosing fresh names to write on all the interior edges.
  \[ %
\beginpgfgraphicnamed{bb-interp1}
\InputIfFileExists{bb-interp1.tikz}{}{\input{./figures/bb-interp1.tikz}}
\endpgfgraphicnamed \quad \Longrightarrow \quad
\beginpgfgraphicnamed{bb-interp2}
\InputIfFileExists{bb-interp2.tikz}{}{\input{./figures/bb-interp2.tikz}}
\endpgfgraphicnamed \]
  Then, write the !-boxes with nesting as depicted in the diagram:
  \[ \term{\ldots[\ldots]C[\ldots[\ldots]B]A} \]
  Write each node in the diagram on the location it occurs (w.r.t. !-boxes):
  \[ \term{\phi.\ldots.[\psi.\ldots.]C[[\psi.\ldots.]B]A} \]
  Finally, add the edges of each node, reading clockwise from the tick. Edges occurring under a clockwise arrow marked $A$ should be enclosed in $\eclock{\dots}{A}$, and edges under an anti-clockwise arrow should be enclosed in $\eanti{\dots}{A}$, where the outermost groups are the ones closest to the node in the picture.
  \[ \term{\phi.+a[<-e]B>A<-d]C.[\psi.+d-c.]C[[\psi.+e-b.]B]A} \]
  The only choices we made in this process were the choice of interior edge names and the order in which to write the individual tensors. However, up to $\equiv$, these are irrelevant. To show that any !-tensor can be represented this way, we simply run the above procedure in reverse.
\end{proof}

Because of this theorem, we use the terms !-tensor and !-graph interchangeably, depending on whether we wish to refer to the syntactic vs. graphical notation.

\section{Instantiating tensor expressions with !-boxes}\label{sec:instantiation}

The following diagram demonstrates two !-box operations we can apply to a graph: killing a !-box is the operation deleting the box $B$ and all contents (including edges to/from $B$), and expanding is the operation creating a new concrete instance of the subgraph inside $B$ (attached appropriately). We can represent the original graph in this diagram with the tensor expression $\term{[\phi.+a-c+b.\psi.+c-d.]B\xi.[-a>B.\zeta.<-b+d]B-e.}$.
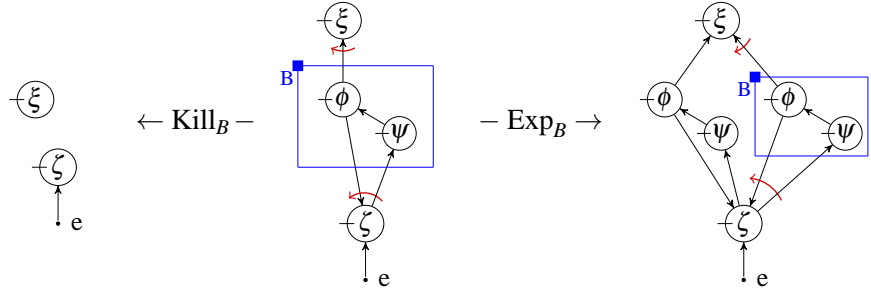
\begin{center}
\beginpgfgraphicnamed{bb-ex1-kill}
\InputIfFileExists{bb-ex1-kill.tikz}{}{\input{./figures/bb-ex1-kill.tikz}}
\endpgfgraphicnamed
	\quad$\leftarrow \Kill_B -$
\beginpgfgraphicnamed{bb-ex1}
\InputIfFileExists{bb-ex1.tikz}{}{\input{./figures/bb-ex1.tikz}}
\endpgfgraphicnamed
	\quad$-\, \Exp_B \rightarrow$\quad
\beginpgfgraphicnamed{bb-ex1-exp}
\InputIfFileExists{bb-ex1-exp.tikz}{}{\input{./figures/bb-ex1-exp.tikz}}
\endpgfgraphicnamed
\end{center}

We can define both of these operations formally. Since expansion involves copying various edge/!-box names, we need a means of obtaining fresh names. Let $\Edges(G) \subset \edgenames$ and $\Boxes(G) \subset \boxnames$ be the edge names and !-box names occurring in a !-tensor $G$, respectively.

\begin{definition}
	A \textit{freshness function} for a !-tensor $G$ is a pair of bijections $\fr: \edgenames \to \edgenames$ and $\fr: \boxnames \to \boxnames$ such that
  \[ \Edges(G) \cap \fr(\Edges(G)) = \varnothing \quad \textrm{and} \quad
     \Boxes(G) \cap \fr(\Boxes(G)) = \varnothing \]
\end{definition}

For !-tensor expressions $G$ or edgeterms $e$, we will write $\fr(G)$ or $\fr(e)$ to designate the new expression with names substituted according to the given bijections.

\begin{definition}
  We define $\Op_B \in \{ \Exp_B, \Kill_B \}$ recursively over !-tensor expressions. For most cases, both operations act trivially:
  \begin{align*}
    \Op_B(G H) & := \Op_B(G) \Op_B(H) &
    \Op_B(e f) &:= \Op_B(e) \Op_B(f) \\
    \Op_B(\term{[G]A}) &:= \term{[\Op_B(G)]A} &
    \Op_B(\term{[e>A}) &:= \term{[\Op_B(e)>A} \\
    \Op_B(\phi_e) &:= \phi_{\Op_B(e)} &
    \Op_B(\term{<e]A}) &:= \term{<\Op_B(e)]A} \\
    \Op_B(x) & : = x & &
  \end{align*}
  where $A \neq B$ and $x \in \{ 1, \term{1.+a-b.}, \ein a, \eout a, \epsilon \}$. Then, for the final three cases:
  \begin{align*}
    \Exp_B(\term{[G]B})  &:= \term{[G]B\fr(G)} &
    \Kill_B(\term{[G]B}) &:= 1 \\
    \Exp_B(\term{[e>B})  &:= \term{[e>B\fr(e)} &
    \Kill_B(\term{[e>B}) &:= \epsilon \\
    \Exp_B(\term{<e]B})  &:= \term{\fr(e)<e]B} &
    \Kill_B(\term{<e]B}) &:= \epsilon
  \end{align*}
\end{definition}

Note that $\Exp_B(G)$ implicitly takes a freshness function as input. If we wish to make this explicit, we will write $\Exp_{B,\fr}$. The above operations can be lifted from !-tensor expressions to !-tensors, i.e. $\equiv$-classes of expressions, because of the following theorem.

\begin{theorem}
	Let $\fr$ be a freshness function for two !-tensor expressions $G, H$. Then $G\equiv H$ implies $\Exp_{B,\fr}(G) \equiv \Exp_{B,\fr}(H)$ and $\Kill_B(G) \equiv \Kill_B(H)$.
\end{theorem}

\begin{proof}
  (Sketch) This can be shown by induction over the structure of !-tensor expressions. It is crucial that we use the \textit{same} freshness function $\fr$ for the expansions of $G$ and $H$, otherwise $G$ and $H$ could end up with distinct free edges or !-boxes.
\end{proof}

These two !-box operations give us a means to define the set of all (concrete) tensors that a single !-tensor represents.

\begin{definition}
  A tensor $G'$ is a \textit{concrete instance} of a !-tensor $G$ if it is obtained from $G$ by repeatedly applying the two !-box operations $\Exp$ and $\Kill$ until $G'$ contains no !-boxes. This sequence of operations is called the \textit{instantiation} of $G'$. We write $\llbracket G \rrbracket$ for the set of all concrete instances of $G$.
\end{definition}

\noindent When we fix a model in some category $\mathcal C$, concrete tensors can then be interpreted as morphisms in $\mathcal C$, just as before. We therefore interpret each !-tensor expression as a family of morphisms in $\mathcal C$, namely the interpretations of each of its concrete instances.

\section{Reasoning with !-boxes}\label{sec:box-reasoning}

The real power of !-boxes comes from the ability to do equational reasoning using infinite families of rules. Just as it makes sense to instantiate a single !-tensor, it makes sense to instantiate an \textit{equation} $G = H$ between two !-tensors, provided they have compatible boundaries.

\begin{definition}
  A \textit{!-tensor equation} `$G = H$' consists of a pair of !-tensors $(G, H)$ that have \textit{compatible boundaries}. That is, they have identical free edge names and !-boxes, $A \prec_G B \Leftrightarrow A \prec_H B$ for all !-boxes in $G$ and $H$, and $\ctx_G(a) = \ctx_H(a)$ for all free edge names.
\end{definition}

Intuitively, we require that the LHS and RHS of a !-tensor equation have the same interface to attach to other graphs (same free variables and same box structure). These consistency conditions guarantee that (i) applying !-box operations to valid equations yields valid equations, and (ii) when $G$ occurs as a sub-expression of some other !-tensor $K$, it can be substituted for $H$ to yield another valid !-tensor $K'$.

\begin{theorem}
  Let $\fr$ be a freshness function for !-tensors $G, H$. Then, if $G = H$ is a !-tensor equation, then so too are:
  \begin{align*}
    \Kill_B(G = H) & \ \ :=\ \  (\Kill_B(G) = \Kill_B(H)) \\
    \Exp_B(G = H)  & \ \ :=\ \  (\Exp_{B,\fr}(G) = \Exp_{B,\fr}(H))
  \end{align*}
\end{theorem}
\begin{proof} 
  (Sketch) It is straightforward to show that killing/expanding $B$ affects the free variables and the !-boxes in the same way on the LHS/RHS. To check the contexts, split a single free edge name into 3 cases, depending on whether the !-box $B$ occurs in the node-context of $a$, the edge context of $a$, or neither. In all cases, $a$ and/or $\fr(a)$ will have identical contexts on the LHS/RHS.
\end{proof}

As in the case of !-tensors, we can define $\llbracket G = H \rrbracket$ to be the set of all concrete rules derivable from $G = H$ using the !-box operations. A valid model of a graphical theory is then one where all of the equations in $\llbracket G = H \rrbracket$ hold for each equation $G = H$. Proving that a rule holds for \textit{all} of its instances could be a daunting task in general, however in many cases a technique called \textit{!-box induction}---which we will meet shortly---comes to the rescue.


We obtain a notion of substitution of sub-expressions constructively, via inference rules. The first few should look familiar as congruence- and substitution-like rules for !-tensors.
\[
	\scalebox{1.0}{
  \AxiomC{$G = H$}
  \RightLabel{\scriptsize(Prod)}
  \UnaryInfC{$GK = HK$}
  \DisplayProof
\qquad
  \AxiomC{$G=H$}
  \RightLabel{\scriptsize(Box)}
  \UnaryInfC{$\term{[G]A} = \term{[H]A}$}
  \DisplayProof
\qquad
  \AxiomC{$G = H$}
  \RightLabel{\scriptsize(Rename)}
  \UnaryInfC{$G[a\to b] = H[a\to b]$}
  \DisplayProof
  }
\]
Where $G[a\to b]$ and $H[a\to b]$ are $G$ and $H$ with the free edge/!-box name $a$ replaced by $b$. We require that $K$ and $A$ are chosen such that $GK$, $HK$, $\term{[G]A}$, and $\term{[H]A}$ are well-defined. These rules provide the conditions under which some equation $G = H$ can be unified, given some context, with a bigger equation $G' = H'$. The final inference rule (Weaken) is less intuitive from the point of view of terms, and is best understood graphically. Consider the following embedding of !-graphs:
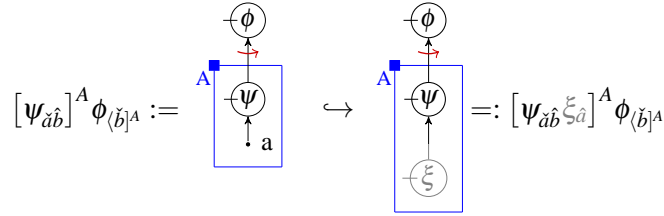
\[ \term{[\psi.-a+b.]A\phi.<-b]A.} :=
\beginpgfgraphicnamed{embed-ex}
\InputIfFileExists{embed-ex.tikz}{}{\input{./figures/embed-ex.tikz}}
\endpgfgraphicnamed =:
   \term{[\psi.-a+b.\color{black!50!white}\xi.+a.\color{black}]A{\phi.<-b]A.}} \]
The LHS does not embed as a sub-term of the RHS, because the !-box $A$ contains more stuff on the RHS. However, semantically, this is perfectly fine, as all of the concrete instances of the LHS will have (uniquely-determined) embeddings into all of the concrete instances of the RHS. So, we also need a rule that allows us to `weaken' !-boxes by adding more nodes to them.
\[
	\scalebox{1.0}{
  \AxiomC{$G=G'$}
  \RightLabel{\scriptsize(Weaken)}
  \UnaryInfC{$\Weak_{A\mapsto K}(G)=\Weak_{A\mapsto K}(G')$}
  \DisplayProof
	}
\]
Where $\Weak_{A\mapsto K}(G)$ is defined recursively as:
\begin{align*}
  \Weak_{A\mapsto K}(\term{[G]A}) &:= \term{[GK]A} \\
  \Weak_{A\mapsto K}(\term{[G]B}) &:= \term{[\Weak_{A\mapsto K}(G)]B} & \text{if } A\not=B \\
  \Weak_{A\mapsto K}(G H) &:= \Weak_{A\mapsto K}(G) \Weak_{A\mapsto K}(H) \\
  \Weak_{A\mapsto K}(x) &:= x & x \in \{ 1, \term{1.+a-b.}, \phi_e \}
\end{align*}
These four rules give us everything we need to apply equations on !-tensors to obtain new equations. For example, the equation on the left below can be applied to delete all of the $\psi$-nodes occurring as input to a $\phi$. An example application of this rule is shown on the right.
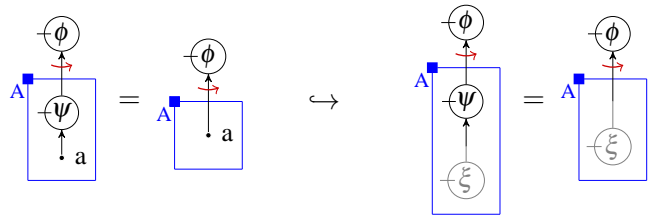
\[ %
\beginpgfgraphicnamed{embed-rule-ex}
\InputIfFileExists{embed-rule-ex.tikz}{}{\input{./figures/embed-rule-ex.tikz}}
\endpgfgraphicnamed \qquad \hookrightarrow \qquad
\beginpgfgraphicnamed{embed-big-rule-ex}
\InputIfFileExists{embed-big-rule-ex.tikz}{}{\input{./figures/embed-big-rule-ex.tikz}}
\endpgfgraphicnamed  \]
We can also add inference rules for each of our !-box operations i.e.
\[
	\scalebox{1.0}{
  \AxiomC{$G = H$}
  \RightLabel{\scriptsize(Exp)}
  \UnaryInfC{$\Exp_B(G = H)$}
  \DisplayProof
\qquad
  \AxiomC{$G=H$}
  \RightLabel{\scriptsize(Kill)}
  \UnaryInfC{$\Kill_B(G = H)$}
  \DisplayProof
  }
\]
Perhaps most interestingly, we can introduce new !-boxes, where previously there were none, via \textit{!-box induction}.
\[
	\scalebox{1.0}{
  \AxiomC{$\Kill_A(G=H)$}
  \AxiomC{$G=H\Rightarrow\Exp_A(G=H)$}
  \RightLabel{\scriptsize(Induction)}
  \BinaryInfC{$G=H$}
  \DisplayProof
  }
\]

As mentioned in Section~\ref{sec:intro}, non-commutative nodes give us the ability to make recursive definitions of variable-arity generators in terms of fixed-arity generators of our theory. This induction principle in turn gives us the means to lift rules about fixed arity generators up to more powerful !-tensor rules.

We conclude by showing a simple example. Suppose we take the theory of a monoid, i.e. the pair of generators \big(%
\beginpgfgraphicnamed{Induc-Unit}
\InputIfFileExists{Induc-Unit.tikz}{}{\input{./figures/Induc-Unit.tikz}}
\endpgfgraphicnamed,%
\beginpgfgraphicnamed{Induc-Multiply}
\InputIfFileExists{Induc-Multiply.tikz}{}{\input{./figures/Induc-Multiply.tikz}}
\endpgfgraphicnamed\big) satisfying the commutativity and unit laws from (\ref{eqn:frob-laws}).Then we can recursively define, as a new generator, an $n$-fold tree of multiplications.
\begin{equation}
\beginpgfgraphicnamed{Induc-SpiderBase}
\InputIfFileExists{Induc-SpiderBase.tikz}{}{\input{./figures/Induc-SpiderBase.tikz}}
\endpgfgraphicnamed \qquad\qquad %
\beginpgfgraphicnamed{Induc-SpiderRecurs}
\InputIfFileExists{Induc-SpiderRecurs.tikz}{}{\input{./figures/Induc-SpiderRecurs.tikz}}
\endpgfgraphicnamed
\end{equation}

\begin{remark}
  Note how non-commutative !-boxes make such recursive definitions possible in the first place, without assuming \textit{a priori} that the family of graphs generated by the definition are symmetric on their inputs/outputs. This need not be true, even in the case where all of the concrete generators are commutative. This limitation in the case of commutative !-boxes was highlighted in~\cite{MerryThesis}, where only a partial proof of the spider theorem for commutative Frobenius algebras could be done using (commutative) !-box induction.
\end{remark}

The first property we would like to prove about such trees is that adjacent trees merge to form bigger trees. As a !-box rule, it looks like this:
\ctikzfig{Induc-MergeLemma}
We can then hit this rule with the induction on $B$ to break it into cases:
\begin{equation}\tag{base}
\beginpgfgraphicnamed{Induc-MergeBase-eq}
\InputIfFileExists{Induc-MergeBase-eq.tikz}{}{\input{./figures/Induc-MergeBase-eq.tikz}}
\endpgfgraphicnamed
\end{equation}
\begin{equation}\tag{step}
\beginpgfgraphicnamed{Induc-MergeLemma}
\InputIfFileExists{Induc-MergeLemma.tikz}{}{\input{./figures/Induc-MergeLemma.tikz}}
\endpgfgraphicnamed \ \ \Rightarrow\ \ %
\beginpgfgraphicnamed{Induc-MergeStep-eq}
\InputIfFileExists{Induc-MergeStep-eq.tikz}{}{\input{./figures/Induc-MergeStep-eq.tikz}}
\endpgfgraphicnamed
\end{equation}
...each of which have simple rewriting proofs:
\ctikzfig{Induc-MergeBase}
\ctikzfig{Induc-MergeInduc}
One caveat is that when we apply the induction hypothesis in step 4, the !-box $B$ must be `fixed' (i.e. we're not allowed to do any instantiation of $B$ via $\Exp$, $\Kill$, etc.). This is because $B$ occurs free on both sides of the implication $G = H \Rightarrow \Exp_B(G = H)$. See~\cite{MerryThesis} for details.

This style of proof is the main workhorse of soundness proofs of rules like the merging rule (a.k.a. `spider rule') for commutative Frobenius algebras described in Section~\ref{sec:intro}, and can be extended to the non-commutative case, proving a similar rule for e.g. \textit{symmetric} Frobenius algebras, giving a purely diagrammatic characterisation of the normal forms described in~\cite{LaudaPfeiffer}.

\if\showproofs1

\appendix

\section{Other !-box operations}

  \begin{definition}
    A !-box operation for box $B$ is the combination of an operation $\Op_B:\graphterms\to\graphterms$ satisfying:
    \begin{align*}
      \Op_B(G_1 G_2) & := \Op_B(G_1) \Op_B(G_2) \\
      \Op_B(\term{[G]A}) &:= \term{[\Op_B(G)]A} \quad \text{if } A\not=B \\
      \Op_B(\phi_e) &:= \phi_{\Op_B(e)} \\
      \Op_B(1) &:= 1
    \end{align*}
    and an operation from $\Op_B:\edgeterms\to\edgeterms$ satisfying:
    \begin{align*}
      \Op_B(e_1 e_2) &:= \Op_B(e_1) \Op_B(e_2)\\
      \Op_B(\term{[e>A}) &:= \term{[\Op_B(e)>A} \quad \text{if } A\not=B \\
      \Op_B(\term{<e]A}) &:= \term{<\Op_B(e)]A} \quad \text{if } A\not=B \\
      \Op_B(a) &:= a \\
      \Op_B(\epsilon) &:= \epsilon
    \end{align*}
  \end{definition}

  We can now define each !-box operation (using appropriate freshness functions) simply by how it affects it's !-box and edge group.

  \begin{definition}
    $\Kill_B$ is the !-box operation satisfying:
    \begin{align*}
      \Kill_B(\term{[G]B}) &:= 1 \\
      \Kill_B(\term{[e>B}) &:= \epsilon \\
      \Kill_B(\term{<e]B}) &:= \epsilon
    \end{align*}
  \end{definition}

  \begin{definition}
    $\Drop_B$ is the !-box operation satisfying:
    \begin{align*}
      \Drop_B(\term{[G]B}) &:= G \\
      \Drop_B(\term{[e>B}) &:= e \\
      \Drop_B(\term{<e]B}) &:= e
    \end{align*}
  \end{definition}

  \begin{definition}
    $\Exp_B$ is the !-box operation satisfying:
    \begin{align*}
      \Exp_B(\term{[G]B}) &:= \term{[G]B\fr(G)} \\
      \Exp_B(\term{[e>B}) &:= \term{[e>B\fr(e)} \\
      \Exp_B(\term{<e]B}) &:= \term{\fr(e)<e]B}
    \end{align*}
  \end{definition}

  \begin{definition}
    $\Copy_B$ is the !-box operation satisfying:
    \begin{align*}
      \Copy_B(\term{[G]B}) &:= \term{[G]B[\fr(G)]{\fr(B)}} \\
      \Copy_B(\term{[e>B}) &:= \term{[e>B[\fr(e)>{\fr(B)}} \\
      \Copy_B(\term{<e]B}) &:= \term{<\fr(e)]{\fr(B)}<e]B}
    \end{align*}
  \end{definition}

  For appropriate freshness functions we can rewrite $\Exp_A$ as $\Drop_A\circ\Copy_A$.

\section{Proofs}

  \begin{lemma}
    \label{lem:contexts}
    If $\ectx_G(a)=[E_1,\ldots,E_n]$, $\nctx_G(a)=[N_1,\ldots,N_m]$ then the following shows contexts affected by operations (writing $B'$ for $\fr(B)$):
    \begin{table}[ht]
    \renewcommand{\arraystretch}{1.2}
    \begin{tabular} {|r|l|c|c|} \hline
                    &            & $\ectx$ & $\nctx$ \\ \hline
      $\Drop_{E_i}$ &        $a$ & $[E_1,\ldots,E_{i-1},E_{i+1},\ldots,E_n]$ &  \\ \hline
      $\Drop_{B_i}$ &        $a$ &  & $[N_1,\ldots,N_{i-1},N_{i+1},\ldots,N_m]$ \\ \hline
      $\Exp_{E_i}$  &   $\!\fr(a)\!$ & $[E_1',\ldots,E_{i-1}',E_{i+1},\ldots,E_n]$ &  \\ \hline
      $\Exp_{B_i}$  &   $\!\fr(a)\!$ & $[E_1',\ldots,E_n']$ & $[N_1',\ldots,N_{i-1}',N_{i+1},\ldots,N_m]$ \\ \hline
      $\Copy_{E_i}$ &   $\!\fr(a)\!$ & $[E_1',\ldots,E_i',E_{i+1},\ldots,E_n]$ & \\ \hline
      $\Copy_{B_i}$ &   $\!\fr(a)\!$ & $[E_1',\ldots,E_n']$ & $[N_1',\ldots,N_i',N_{i+1},\ldots,N_m]$ \\ \hline
    \end{tabular}
    \end{table}
  \end{lemma}

  \begin{theorem}
    $G\in\graphterms \implies \Kill_A(G)\in\graphterms$.
  \end{theorem}
  \begin{proof} (Sketch)
    It is clear that $\Kill_A(G)$ is still constructed as a !-pretensor so we need only check the conditions F1-2, C1-3:
    \begin{itemize}
        \item[F1-2:] Trivial since we have only removed edges/boxes.
        \item[C1:] If $a\in\Edges(\Kill_A(G))$ then contexts were not affected by $\Kill_A$ so same condition holds.
        \item[C2:] If $a\in\Edges(\Kill_A(G))$ with edge context $[E_1,\ldots,E_n]\not=[]$ \\
          then $a\in \Edges(G)$ with edge context $[E_1,\ldots,E_n]$ hence $E_i\prec_G E_{i+1}$. \\
          $\therefore$ $E_i\prec_{\Kill_A(G)} E_{i+1}$
        \item[C3:] $\eout{a},\ein{a}\in\Edges(\Kill_A(G))$ then $\eout{a},\ein{a}\in\Edges(G)$ and $\ectx, \nctx$ were not affected by $\Kill_A$ so the condition still holds.
    \end{itemize}
  \end{proof}

  \begin{theorem}
    $G\in\graphterms \implies \Drop_A(G)\in\graphterms$.
  \end{theorem}
  \begin{proof} (Sketch)
    Again we need only check the conditions F1-2, C1-3:
    \begin{itemize}
        \item[F1-2:] Trivial since we have only removed boxes.
        \item[C1:] Again trivial since $\ectx$ and $\nctx$ have only lost boxes.
        \item[C2:] If $a\in\Edges(\Drop_A(G))$ has edge context $[E_1,\ldots,E_n]\not=[]$ \\
          then $a\in \Edges(G)$ would have edge context $[E_1,\ldots,A,\ldots,E_n]$ \\
          hence $E_1\prec_G\ldots\prec_G A \prec_G\ldots\prec_G E_n$. \\
          $\therefore$ $E_1\prec_{\Drop_A(G)}\ldots\prec_{\Drop_A(G)} E_n$
        \item[C3:] $\eout{a},\ein{a}\in\Edges(\Drop_A(G))$ then $\eout{a},\ein{a}\in\Edges(G)$ and $\ectx, \nctx$ only lost the box $A$ so the condition still holds by removing $A$ from $es$,$bs$.
    \end{itemize}
  \end{proof}

  \begin{theorem}
    $G\in\graphterms \implies \Exp_A(G)\in\graphterms$.
  \end{theorem}
  \begin{proof} (Sketch)
    Again we need only check the conditions F1-2, C1-3:
    \begin{itemize}
      \item[F1-2:]  Trivial since any new edges/boxes are fresh ones which can only have come from the distinct edges/boxes in $G$.
      \item[C1:] Edges in $\Exp_A(G)$ are either edges from $G$ or fresh names for edges in $G$. For the former $\ectx, \nctx$ have not been changed so the condition holds. For the later we know from Lemma \ref{lem:contexts} that boxes in $\ectx, \nctx$ have only been removed or replaced by fresh versions of themselves hence still distinct.
      \item[C2:] For $a\in\Edges(G)$ (writing $B'$ for $\fr(B)$):
        \begin{itemize}
          \item[$\bullet$] $\ectx_{\Exp_A(G)}(a) = [E_1,\ldots,E_n]=\ectx_G(a)$ and $E_i\prec_G E_{i+1} \implies E_i\prec_{\Exp_A(G)} E_{i+1}$ so the condition holds.
          \item[$\bullet$] $\ectx_{\Exp_A(G)}(\fr(a)) = [E_1',\ldots,E_i',E_{i+1},\ldots,E_n]$ for some $E_i$. \\
          $E_i'$ must have been created inside $E_{i+1}$ in $\Exp_A(G)$. \\
          For others it is clear that $j<i \implies E_j'\prec_{\Exp_A(G)} E_{j+1}'$, \\
          $j>i \implies E_j\prec_{\Exp_A(G)} E_{j+1}$.
        \end{itemize}
      \item[C3:] For $\eout{a},\ein{a}\in\Edges(\Exp_A(G))$ s.t $\eout{a},\ein{a}\in\Edges(G)$ then $\ectx, \nctx$ were not affected by $\Exp_A$ so the condition still holds. \\
      For $\eout{\fr(a)},\ein{\fr(a)}\in\Edges(\Exp_A(G))$ s.t $\eout{a},\ein{a}\in\Edges(G)$ \\
      then $\exists E_i,N_i \in \boxnames$ s.t \\
      $[E_1,\ldots,E_n].\nctx_G(\ein{a})=\ectx_G(\eout{a}).[N_1,\ldots,N_m]$\\
      $[E_1,\ldots,E_n].\nctx_G(\eout{a})=\ectx_G(\ein{a}).[N_1,\ldots,N_m]$\\
      We have four cases to check for how $\ectx,\nctx$ are affected by $\Exp_A$:
      \begin{itemize}
        \item[$\bullet$] $A\in \ectx_G(\ein{a})\cap \ectx_G(\eout{a})$ then $A=E_i$ for some $i$ and letting $es:=[E_1',\ldots,E_{i-1}',E_{i+1},\ldots,E_n]$ and $bs:=[N_1,\ldots,N_m]$ we have the condition holds.
        \item[$\bullet$] $A\in \nctx_G(\ein{a})\cap \nctx_G(\eout{a})$ then $A=N_i$ for some $i$ and letting $es:=[E_1',\ldots,E_m']$ and $bs:=[N_1',\ldots,N_{i-1}',N_{i+1},\ldots,N_n]$ we have the condition holds.
        \item[$\bullet$] $A\in \ectx_G(\ein{a})\cap \nctx_G(\eout{a})$ then letting $es:=[E_1',\ldots,E_m']$ and $bs:=[N_1,\ldots,N_m]$ we have the condition holds.
        \item[$\bullet$] $A\in \nctx_G(\ein{a})\cap \ectx_G(\eout{a})$ similar to previous case.
      \end{itemize}
    \end{itemize}
  \end{proof}

  \begin{theorem}
    $G\in\graphterms \implies \Copy_A(G)\in\graphterms$.
  \end{theorem}
  \begin{proof} (Sketch)
    Again we need only check the conditions F1-2, C1-3:
    \begin{itemize}
      \item[F1-2:]  Trivial since any new edges/boxes are fresh ones which can only have come from the distinct edges/boxes in $G$.
      \item[C1:] Edges in $\Copy_A(G)$ are either edges from $G$ or fresh names for edges in $G$. For the former $\ectx, \nctx$ have not been changed so the condition holds. For the later we know from Lemma \ref{lem:contexts} that boxes in $\ectx, \nctx$ have only been replaced by fresh versions of themselves hence still distinct.
      \item[C2:] For $a\in\Edges(G)$ (writing $B'$ for $\fr(B)$):
        \begin{itemize}
          \item[$\bullet$] $\ectx_{\Copy_A(G)}(a) = [E_1,\ldots,E_n]=\ectx_G(a)$ and $E_i\prec_G E_{i+1} \implies E_i\prec_{\Copy_A(G)} E_{i+1}$ so the condition holds.
          \item[$\bullet$] $\ectx_{\Copy_A(G)}(\fr(a)) = [E_1',\ldots,E_i',E_{i+1},\ldots,E_n]$ for some $E_i$. \\
          $E_i'$ must have been created inside $E_{i+1}$ in $\Copy_A(G)$. \\
          For others it is clear that $j<i \implies E_j'\prec_{\Copy_A(G)} E_{j+1}'$, \\
          $j>i \implies E_j\prec_{\Copy_A(G)} E_{j+1}$.
        \end{itemize}
      \item[C3:] For $\eout{a},\ein{a}\in\Edges(\Copy_A(G))$ s.t $\eout{a},\ein{a}\in\Edges(G)$ then $\ectx, \nctx$ were not affected by $\Copy_A$ so the condition still holds. \\
      For $\eout{\fr(a)},\ein{\fr(a)}\in\Edges(\Copy_A(G))$ s.t $\eout{a},\ein{a}\in\Edges(G)$ $\exists E_i,N_i \in \boxnames$ s.t \\
      $[E_1,\ldots,E_n].\nctx_G(\ein{a})=\ectx_G(\eout{a}).[N_1,\ldots,N_m]$\\
      $[E_1,\ldots,E_n].\nctx_G(\eout{a})=\ectx_G(\ein{a}).[N_1,\ldots,N_m]$\\
      We have four cases to check for how $\ectx,\nctx$ are affected by $\Copy_A$:
      \begin{itemize}
        \item[$\bullet$] $A\in \ectx_G(\ein{a})\cap \ectx_G(\eout{a})$ then $A=E_i$ for some $i$ and letting $es:=[E_1',\ldots,E_i',E_{i+1},\ldots,E_n]$ and $bs:=[N_1,\ldots,N_m]$ we have the condition holds.
        \item[$\bullet$] $A\in \nctx_G(\ein{a})\cap \nctx_G(\eout{a})$ then $A=N_i$ for some $i$ and letting $es:=[E_1',\ldots,E_m']$ and $bs:=[N_1',\ldots,N_i',N_{i+1},\ldots,N_n]$ we have the condition holds.
        \item[$\bullet$] $A\in \ectx_G(\ein{a})\cap \nctx_G(\eout{a})$ then letting $es:=[E_1',\ldots,E_m']$ and $bs:=[N_1,\ldots,N_m]$ we have the condition holds.
        \item[$\bullet$] $A\in \nctx_G(\ein{a})\cap \ectx_G(\eout{a})$ similar to previous case.
      \end{itemize}
    \end{itemize}
  \end{proof}
  \begin{theorem}
    $G\in\graphterms \implies \Op_A(G)\in\graphterms$ \qquad for $\Op$ any one of our four defined !-box operations.
  \end{theorem}
  \begin{proof} 
    (Sketch) Shown by checking that the conditions F1-2, C1-3 still hold. F1-2 are trivial since new edges/boxes have new names. C1-3 can each be checked using Lemma \ref{lem:contexts}.
  \end{proof}

  \begin{theorem}
    If in graph $G$, box $B$ is nested inside box $A$ ($B\prec_G\ldots\prec_G A$) then we can find a freshness function $\fr'$ s.t. the following equations hold on $G$:
    \begin{align*}
      \Kill_{A} \circ \Op_{B,\fr_B} &= \Kill_{A} \\
      \Drop_{A} \circ \Op_{B,\fr_B} &= \Op_{B,\fr_B} \circ \Drop_{A} \\
      \Exp_{A,\fr_A} \circ \Op_{B,\fr_B} &= \Op_{B',\fr'} \circ \Op_{B,\fr_B} \circ \Exp_{A,\fr_A} \\
      \Copy_{A,\fr_A} \circ \Op_{B,\fr_B} &= \Op_{B',\fr'} \circ \Op_{B,\fr_B} \circ \Copy_{A,\fr_A}
    \end{align*}
  \end{theorem}
  \begin{proof} 
    (Sketch) Induction on $\edgeterms$ and $\graphterms$. The interesting cases are $\term{[e>C}$, $\term{<e]C}$ and $\term{[G]C}$ where $C=A,B$. These can be checked explicitly. The freshness function $\fr'$ only requires that $\fr_A(x)\mapsto\fr_A(\fr_B(x))$.
  \end{proof}

	We also expect that any of our defined !-box operations preserves equality. For $\Op\in\{\Kill,\Drop,\Exp,\Copy\}$:

  \begin{prooftree}
    \AxiomC{$G=H$}
    \RightLabel{\scriptsize($\Op_A$)}
    \UnaryInfC{$\Op_A(G=H)$}
  \end{prooftree}

  We can use these rules to produce more rules such as a second product rule:
  \begin{theorem}\qquad
    \begin{prooftree}
      \AxiomC{$G=G'$}
      \RightLabel{\scriptsize(prod2)}
      \UnaryInfC{$HG=HG'$}
    \end{prooftree}
  \end{theorem}
  \begin{proof}
    \begin{prooftree}
      \AxiomC{$G=G'$}
      \RightLabel{\scriptsize(prod)}
      \UnaryInfC{$GH=G'H$}
      \RightLabel{\scriptsize(equivalence)}
      \UnaryInfC{$HG=HG'$}
    \end{prooftree}
  \end{proof}

  We can also derive a stronger weakening rule:
  \begin{theorem}\qquad
    \begin{prooftree}
      \AxiomC{$G=G'$}
      \AxiomC{$K=K'$}
      \RightLabel{\scriptsize(Weak')}
      \BinaryInfC{$\Weak_{A\to K}(G)=\Weak_{A\to K'}(G')$}
    \end{prooftree}
  \end{theorem}
  \begin{proof}
    Proof by induction on the complexity of $G'$ that
    \begin{prooftree}
      \AxiomC{$K=K'$}
      \RightLabel{\scriptsize}
      \UnaryInfC{$\Weak_{A\to K}(G')=\Weak_{A\to K'}(G')$}
    \end{prooftree} then transitivity gives the desired result.
    \begin{itemize}
      \item $G'=\phi_e,1$ trivial
      \item $G'=\gbox{H}{B}$ with $B\not=A$
      \begin{prooftree}
        \AxiomC{$K=K'$}
        \RightLabel{\scriptsize(IH)}
        \UnaryInfC{$\Weak_{A\to K}(H)=\Weak_{A\to K'}(H)$}
        \RightLabel{\scriptsize(Box)}
        \UnaryInfC{$\gbox{\Weak_{A\to K}(H)}{B}=\gbox{\Weak_{A\to K'}(H)}{B}$}
        \RightLabel{\scriptsize($\equiv$)}
        \UnaryInfC{$\Weak_{A\to K}(\gbox{H}{B})=\Weak_{A\to K'}(\gbox{H}{B})$}
      \end{prooftree}
      \item $G'=\term{[H]A}$
      \begin{prooftree}
        \AxiomC{$K=K'$}
        \RightLabel{\scriptsize(prod)}
        \UnaryInfC{$HK=HK'$}
        \RightLabel{\scriptsize(Box)}
        \UnaryInfC{$\gbox{HK}{A}=\gbox{HK'}{A}$}
        \RightLabel{\scriptsize($\equiv$)}
        \UnaryInfC{$\Weak_{A\to K}(\term{[H]A})=\Weak_{A\to K'}(\term{[H]A})$}
      \end{prooftree}
      \item $G'=H_1H_2$ \\
      Writing $W_K$ for $\Weak_{A\to K}$:
      \begin{prooftree}
        \AxiomC{$K=K'$}
        \RightLabel{\scriptsize(IH)}
        \UnaryInfC{$W_K(H_1)=W_{K'}(H_1)$}
        \AxiomC{$K=K'$}
        \RightLabel{\scriptsize(IH)}
        \UnaryInfC{$W_K(H_2)=W_{K'}(H_2)$}
        \RightLabel{\scriptsize(prod)}
        \BinaryInfC{$W_K(H_1H_2)=W_{K'}(H_1H_2)$}
      \end{prooftree}
    \end{itemize}
  \end{proof}

	\begin{theorem}
		Let $\fr$ be a freshness function for two !-tensor expressions $G, H$. Then $G\equiv H$ implies $\Op_{B,\fr}(G) \equiv \Op_{B,\fr}(H)$.
	\end{theorem}
  \begin{proof}
    We need to check our enforced equivalences still hold after $\Op_B$:
    \begin{itemize}
      \item $\Op_B(1G)\equiv\Op_B(1)\Op_B(G)\equiv1\Op_B(G)\equiv\Op_B(G)\equiv\ldots\equiv\Op_B(G1)$
      \item $\Op_B(\epsilon e)\equiv\Op_B(\epsilon)\Op_B(e)\equiv\epsilon\Op_B(e)\equiv\Op_B(e)\equiv\ldots\equiv\Op_B(e\epsilon)$
      \item $\term{\Op_B([\epsilon>A)}\equiv\term{[\Op_B(\epsilon)>A}\equiv\term{[\epsilon>A}\equiv\ldots\equiv\term{\Op_B(<\epsilon]A)}$ \quad if $A\not\equiv B$ \\
      When $A\equiv B$ we need to check each operation individually but they are all easy to show
      \item For products:
        \begin{align*}
          \Op_A(G H) &\equiv \Op_A(G)\Op_A(H) \\
          &\equiv \Op_A(H)\Op_A(G) \\
          &\equiv \Op_A(H G)
        \end{align*}
      \item Given $a\in\Bound(G)$ and $b\notin\Edges(G)$
    \end{itemize}
  \end{proof}

  \begin{theorem}
    If $A\not=B$ are not nested inside each other then any !-box operations $\Op_A, \Op'_B$ commute.
  \end{theorem}
  \begin{proof}
    (Sketch) Each operation only affects a graph locally at their respective boxes which has no effect on the other box.
  \end{proof}

  \begin{corollary}
    Any string of operations on a graph $G$ can be written in a normal form where all top level box operations (ordered by !-box name) are applied first then all operations on first level children (ordered by !-box name) and so on.
  \end{corollary}

  \begin{proof}
    (Sketch) Take the first top level box, $A$ in $G$ then use the previous theorems to move operations on $A$ to be first. This only creates finitely many extra operations on boxes nested in $A$. Applying the same to the rest of the top level boxes we get them all first and in order so we can ignore them and recursively normalise the order on what is left.
  \end{proof}

\fi


\vspace{2cm}

{\small
\bibliographystyle{eptcs}
\bibliography{bibfile}
}

\end{document}

%% file: preamble.tex
\usepackage{tikzfig}
\input{tikzstyles.tex}

\usepackage{noncommgraph}

\newcommand{\showproofs}{1}

\usepackage{amsmath}
\usepackage{amssymb}
\usepackage{stmaryrd}
\usepackage{amsthm}
\usepackage{xspace}
\usepackage{hyperref}
\usepackage{graphicx}
\graphicspath{{figures/}}

\usepackage{algpseudocode}
\usepackage{algorithm}
\algblock[Switch]{Switch}{EndSwitch}
\algblock[Case]{Case}{EndCase}
\algtext*{EndCase}

\newtheorem{theorem}{Theorem}[section]
\newtheorem*{theorem*}{Theorem}

\newtheorem{lemma}[theorem]{Lemma}
\newtheorem{corollary}[theorem]{Corollary}
\theoremstyle{definition}
\theoremstyle{definition}
\theoremstyle{definition}\newtheorem{definition}[theorem]{Definition}
\theoremstyle{definition}
\theoremstyle{definition}\newtheorem{remark}[theorem]{Remark}
\theoremstyle{definition}
\theoremstyle{definition}
\theoremstyle{definition}
\theoremstyle{definition}
\theoremstyle{definition}

\input{defs.tex}

\def\bR{\begin{color}{red}} 
\def\bB{\begin{color}{blue}}
\def\bM{\begin{color}{magenta}}
\def\bC{\begin{color}{cyan}}
\def\bW{\begin{color}{white}}
\def\bBl{\begin{color}{black}} 
\def\bG{\begin{color}{green}}
\def\bY{\begin{color}{yellow}}
\def\bDG{\begin{color}{green!70!black}}
\def\e{\end{color}}

\usepackage{wrapfig}
  {\gdef\scalefactor{1.0}\begin{center}\proofSkipAmount \leavevmode}%
  {\scalebox{\scalefactor}{\DisplayProof}\proofSkipAmount \end{center} }

%% file: tikzstyles.tex
\tikzstyle{every picture}=[baseline=-0.25em,scale=0.6]
\tikzstyle{dotpic}=[scale=0.6]
\tikzstyle{diredges}=[every to/.style={diredge}]
\tikzstyle{dot graph}=[shorten <=-0.1mm,shorten >=-0.1mm,scale=0.6]
\tikzstyle{digraph}=[-latex]
\tikzstyle{plot point}=[circle,fill=black,minimum width=2mm,inner sep=0]
\tikzstyle{string graph}=[scale=0.6]
\tikzstyle{sg diredge}=[-stealth]
\tikzstyle{rewrite edge}=[-open triangle 45]
\tikzstyle{sg bold diredge}=[-stealth,thick,shorten >=-1pt]
\tikzstyle{sg vertex}=[circle,minimum width=2.2mm,fill=white,draw=black,inner sep=0mm]
\tikzstyle{labelled sg vertex}=[circle,minimum width=7mm,fill=white,draw=black,inner sep=0mm]
\tikzstyle{sg grey vertex}=[sg vertex,fill=gray!30!white]
\tikzstyle{sg black vertex}=[sg vertex,fill=black]
\tikzstyle{sg bold vertex}=[circle,minimum width=2.2mm,fill=white,draw=black,very thick,inner sep=0mm]
\tikzstyle{sg wire vertex}=[circle,minimum width=1mm,fill=black,inner sep=0mm]

\tikzstyle{tick vertex}=[rectangle,fill=black,minimum height=1mm,minimum width=2.5mm,inner sep=0mm]

\tikzstyle{braceedge}=[decorate,decoration={brace,amplitude=2mm,raise=-1mm}]
\tikzstyle{small braceedge}=[decorate,decoration={brace,amplitude=1mm,raise=-1mm}]
\tikzstyle{left hook arrow}=[left hook-latex]
\tikzstyle{right hook arrow}=[right hook-latex]

\tikzstyle{dot}=[inner sep=0.7mm,minimum width=0pt,minimum height=0pt,fill=black,draw=black,shape=circle]
\tikzstyle{white dot}=[dot,fill=white]
\tikzstyle{alt white dot}=[white dot,label={[xshift=2.9mm,yshift=-0.1mm]left:$\cdot$}]
\tikzstyle{gray dot}=[dot,fill=gray!50]
\tikzstyle{box vertex}=[draw=black,rectangle]
\tikzstyle{whitebg}=[fill=white,inner sep=2pt]
\tikzstyle{graph state vertex}=[sg vertex,fill=black]

\tikzstyle{wide point}=[fill=white,draw=black,shape=isosceles triangle,shape border rotate=90,isosceles triangle stretches=true,inner sep=1pt,minimum width=1.5cm,minimum height=5mm]
\tikzstyle{wide copoint}=[fill=white,draw=black,shape=isosceles triangle,shape border rotate=-90,isosceles triangle stretches=true,inner sep=1pt,minimum width=1.5cm,minimum height=5mm]
\tikzstyle{symm}=[ultra thick,shorten <=-1mm,shorten >=-1mm]

\tikzstyle{square box}=[rectangle,fill=white,draw=black,minimum height=6mm,minimum width=6mm]
\tikzstyle{square gray box}=[rectangle,fill=gray!30,draw=black,minimum height=6mm,minimum width=6mm]
\tikzstyle{point}=[regular polygon,regular polygon sides=3,draw=black,scale=0.75,inner sep=-0.5pt,minimum width=7mm,fill=white]
\tikzstyle{copoint}=[point,regular polygon rotate=180,fill=white]
\tikzstyle{gray point}=[point,fill=gray!40!white]
\tikzstyle{gray copoint}=[copoint,fill=gray!40!white]

\tikzstyle{open graph}=[baseline=-0.25em]
\tikzstyle{greybg}=[background rectangle/.style={fill=black!5,draw=black!30,rounded corners=1ex}, show background rectangle]

\tikzstyle{edge point}=[circle,minimum width=1mm,fill=black,inner sep=0mm]
\tikzstyle{vertex point}=[circle,minimum width=2.2mm,fill=white,draw=black,inner sep=0mm]
\tikzstyle{gray vertex point}=[circle,minimum width=2.2mm,fill=gray!30!white,draw=black,inner sep=0mm]
\tikzstyle{edge label}=[inner sep=2pt, font=\small]
\tikzstyle{on edge label}=[fill=white, font=\footnotesize, inner sep=1 pt]

\newcommand{\edgearrow}{{\arrow[black]{>}}}
\newcommand{\edgetick}{{\arrow[black,scale=0.7,very thick]{|}}}
\tikzstyle{diredge}=[postaction=decorate,decoration={markings, mark=at position 0.55 with \edgearrow}]
\tikzstyle{medium diredge}=[postaction=decorate,decoration={markings, mark=at position 0.75 with \edgearrow}]
\tikzstyle{short diredge}=[->]
\tikzstyle{halfedge}=[-)]
\tikzstyle{other halfedge}=[(-]
\tikzstyle{freeedge}=[(-)]
\tikzstyle{white edge}=[line width=5pt,white]
\tikzstyle{tick}=[postaction=decorate,decoration={markings, mark=at position 0.5 with \edgetick}]
\tikzstyle{small map edge}=[|-latex, gray!60!blue, shorten <=0.9mm, shorten >=0.5mm]
\tikzstyle{thick dashed edge}=[very thick,dashed,gray!40]
\tikzstyle{dashed edge}=[densely dotted,thick]
\tikzstyle{map edge}=[|-latex,very thick, gray!40, shorten <=1mm, shorten >=0.5mm]
\tikzstyle{tickedge}=[postaction=decorate,
  decoration={markings, mark=at position 0.5 with \edgetick}]
\tikzstyle{dirtickedge}=[postaction=decorate,
  decoration={markings, mark=at position 0.5 with \edgetick},
  decoration={markings, mark=at position 0.85 with \edgearrow}]
\tikzstyle{dirdoubletickedge}=[postaction=decorate,
  decoration={markings, mark=at position 0.4 with \edgetick},
  decoration={markings, mark=at position 0.6 with \edgetick},
  decoration={markings, mark=at position 0.85 with \edgearrow}]

\tikzstyle{arrs}=[-latex,font=\small,auto]
\tikzstyle{arrow plain}=[arrs]
\tikzstyle{arrow dashed}=[dashed,arrs]
\tikzstyle{arrow bold}=[very thick,arrs]
\tikzstyle{arrow hide}=[draw=white!0,-]
\tikzstyle{arrow reverse}=[latex-]
\tikzstyle{cdnode}=[]

\tikzstyle{cnot}=[fill=white,shape=circle,inner sep=-1.4pt]
\tikzstyle{bang box}=[draw=black,dashed,minimum height=12mm,minimum width=12mm,fill=gray!20]
\tikzstyle{wire label}=[font=\footnotesize, auto]

%% file: defs.tex
\tikzstyle{cdiag}=[matrix of math nodes, row sep=3em, column sep=3em, text height=1.5ex, text depth=0.25ex,inner sep=0.5em]
\tikzstyle{arrow above}=[transform canvas={yshift=0.5ex}]
\tikzstyle{arrow below}=[transform canvas={yshift=-0.5ex}]



%% file: figures/frob_sig.tikz
\begin{tikzpicture}
	\begin{pgfonlayer}{nodelayer}
		\node [style=none] (0) at (-5.75, 0) {$\bigg\{$};
		\node [style=white dot] (1) at (-4, 0) {};
		\node [style=white dot] (2) at (-1, -0.25) {};
		\node [style=wire] (3) at (-4.5, -0.75) {};
		\node [style=wire] (4) at (-3.5, -0.75) {};
		\node [style=wire] (5) at (-4, 0.75) {};
		\node [style=wire] (6) at (-1, 0.75) {};
		\node [style=none] (7) at (-2.5, -0.25) {$,$};
		\node [style=white dot] (8) at (5, 0.25) {};
		\node [style=wire] (9) at (1.5, 0.75) {};
		\node [style=wire] (10) at (2.5, 0.75) {};
		\node [style=none] (11) at (3.5, -0.25) {$,$};
		\node [style=white dot] (12) at (2, 0) {};
		\node [style=wire] (13) at (2, -0.75) {};
		\node [style=wire] (14) at (5, -0.75) {};
		\node [style=none] (15) at (0.5, -0.25) {$,$};
		\node [style=none] (16) at (6.5, 0) {$\bigg\}$};
	\end{pgfonlayer}
	\begin{pgfonlayer}{edgelayer}
		\draw [style=directed, bend right=15, looseness=1.00] (4) to (1);
		\draw [style=directed] (1) to (5);
		\draw [style=directed, bend left=15, looseness=1.00] (3) to (1);
		\draw [style=directed] (2) to (6);
		\draw [style=directed, bend right=15, looseness=1.00] (12) to (10);
		\draw [style=directed] (13) to (12);
		\draw [style=directed, bend left=15, looseness=1.00] (12) to (9);
		\draw [style=directed] (14) to (8);
	\end{pgfonlayer}
\end{tikzpicture}

%% file: figures/comm_frob_eqns.tikz
\begin{tikzpicture}
	\begin{pgfonlayer}{nodelayer}
		\node [style=white dot] (0) at (-10, 2) {};
		\node [style=wire] (1) at (-11, 0.5) {};
		\node [style=wire] (2) at (-10, 2.75) {};
		\node [style=wire] (3) at (-10, 0.5) {};
		\node [style=wire] (4) at (-9, 0.5) {};
		\node [style=white dot] (5) at (-9.5, 1.25) {};
		\node [style=white dot] (6) at (-7.5, 1.25) {};
		\node [style=wire] (7) at (-8, 0.5) {};
		\node [style=wire] (8) at (-6, 0.5) {};
		\node [style=white dot] (9) at (-7, 2) {};
		\node [style=wire] (10) at (-7, 2.75) {};
		\node [style=wire] (11) at (-7, 0.5) {};
		\node [style=none] (12) at (-8.5, 1.5) {$=$};
		\node [style=wire] (13) at (-3.5, 0.75) {};
		\node [style=white dot] (14) at (-4, 1.5) {};
		\node [style=wire] (15) at (-4, 2.25) {};
		\node [style=white dot] (16) at (-4.5, 0.75) {};
		\node [style=wire] (17) at (-2.25, 0.75) {};
		\node [style=wire] (18) at (-2.25, 2.25) {};
		\node [style=none] (19) at (-2.75, 1.5) {$=$};
		\node [style=white dot] (20) at (-7, -2) {};
		\node [style=wire] (21) at (-11, -0.5) {};
		\node [style=wire] (22) at (-7, -2.75) {};
		\node [style=white dot] (23) at (-4.5, -0.75) {};
		\node [style=none] (24) at (-2.75, -1.5) {$=$};
		\node [style=white dot] (25) at (-9.5, -1.25) {};
		\node [style=wire] (26) at (-7, -0.5) {};
		\node [style=white dot] (27) at (-7.5, -1.25) {};
		\node [style=wire] (28) at (-8, -0.5) {};
		\node [style=wire] (29) at (-4, -2.25) {};
		\node [style=none] (30) at (-8.5, -1.5) {$=$};
		\node [style=wire] (31) at (-9, -0.5) {};
		\node [style=wire] (32) at (-10, -0.5) {};
		\node [style=wire] (33) at (-6, -0.5) {};
		\node [style=wire] (34) at (-2.25, -2.25) {};
		\node [style=white dot] (35) at (-4, -1.5) {};
		\node [style=white dot] (36) at (-10, -2) {};
		\node [style=wire] (37) at (-3.5, -0.75) {};
		\node [style=wire] (38) at (-2.25, -0.75) {};
		\node [style=wire] (39) at (-10, -2.75) {};
		\node [style=white dot] (40) at (7.75, 0.5) {};
		\node [style=white dot] (41) at (7.75, -0.5) {};
		\node [style=wire] (42) at (7.25, 1.25) {};
		\node [style=wire] (43) at (8.25, -1.25) {};
		\node [style=wire] (44) at (7.25, -1.25) {};
		\node [style=wire] (45) at (8.25, 1.25) {};
		\node [style=wire] (46) at (5.25, 1.25) {};
		\node [style=wire] (47) at (4.25, -1.25) {};
		\node [style=white dot] (48) at (5.25, 0.5) {};
		\node [style=wire] (49) at (5.75, -1.25) {};
		\node [style=wire] (50) at (3.75, 1.25) {};
		\node [style=white dot] (51) at (4.25, -0.5) {};
		\node [style=none] (52) at (6.5, 0) {$=$};
		\node [style=white dot] (53) at (-0.25, -1.5) {};
		\node [style=wire] (54) at (0.25, -0.75) {};
		\node [style=wire] (55) at (-0.75, -0.75) {};
		\node [style=wire] (56) at (-0.25, -2.25) {};
		\node [style=wire] (57) at (1.5, -0.5) {};
		\node [style=white dot] (58) at (2, -1.5) {};
		\node [style=wire] (59) at (2.5, -0.5) {};
		\node [style=wire] (60) at (2, -2.25) {};
		\node [style=none] (61) at (1, -1.5) {$=$};
		\node [style=wire] (62) at (1.5, 0.5) {};
		\node [style=white dot] (63) at (-0.25, 1.5) {};
		\node [style=white dot] (64) at (2, 1.5) {};
		\node [style=wire] (65) at (2.5, 0.5) {};
		\node [style=wire] (66) at (2, 2.25) {};
		\node [style=wire] (67) at (0.25, 0.75) {};
		\node [style=none] (68) at (1, 1.5) {$=$};
		\node [style=wire] (69) at (-0.75, 0.75) {};
		\node [style=wire] (70) at (-0.25, 2.25) {};
	\end{pgfonlayer}
	\begin{pgfonlayer}{edgelayer}
		\draw [style=directed] (0) to (2);
		\draw [style=directed, bend left=15, looseness=1.00] (1) to (0);
		\draw [style=directed, bend right=15, looseness=1.00] (4) to (5);
		\draw [style=directed, bend left=15, looseness=1.00] (3) to (5);
		\draw [style=directed, bend right=15, looseness=1.00] (5) to (0);
		\draw [style=directed] (9) to (10);
		\draw [style=directed, bend right=15, looseness=1.00] (8) to (9);
		\draw [style=directed, bend left=15, looseness=1.00] (7) to (6);
		\draw [style=directed, bend right=15, looseness=1.00] (11) to (6);
		\draw [style=directed, bend left=15, looseness=1.00] (6) to (9);
		\draw [style=directed] (14) to (15);
		\draw [style=directed, bend right=15, looseness=1.00] (13) to (14);
		\draw [style=directed, bend left=15, looseness=1.00] (16) to (14);
		\draw [style=directed] (17) to (18);
		\draw [style=directed] (39) to (36);
		\draw [style=directed, bend left=15, looseness=1.00] (36) to (21);
		\draw [style=directed, bend right=15, looseness=1.00] (25) to (31);
		\draw [style=directed, bend left=15, looseness=1.00] (25) to (32);
		\draw [style=directed, bend right=15, looseness=1.00] (36) to (25);
		\draw [style=directed] (22) to (20);
		\draw [style=directed, bend right=15, looseness=1.00] (20) to (33);
		\draw [style=directed, bend left=15, looseness=1.00] (27) to (28);
		\draw [style=directed, bend right=15, looseness=1.00] (27) to (26);
		\draw [style=directed, bend left=15, looseness=1.00] (20) to (27);
		\draw [style=directed] (29) to (35);
		\draw [style=directed, bend right=15, looseness=1.00] (35) to (37);
		\draw [style=directed, bend left=15, looseness=1.00] (35) to (23);
		\draw [style=directed] (34) to (38);
		\draw [style=directed] (43) to (41);
		\draw [style=directed] (41) to (40);
		\draw [style=directed] (40) to (45);
		\draw [style=directed] (44) to (41);
		\draw [style=directed] (40) to (42);
		\draw [style=directed] (47) to (51);
		\draw [style=directed] (51) to (48);
		\draw [style=directed] (48) to (46);
		\draw [style=directed, bend left=15, looseness=1.00] (51) to (50);
		\draw [style=directed, bend right=15, looseness=1.00] (49) to (48);
		\draw [style=directed, bend left=15, looseness=1.00] (53) to (55);
		\draw [style=directed, bend right=15, looseness=1.00] (53) to (54);
		\draw [style=directed] (56) to (53);
		\draw [style=directed, in=-90, out=135, looseness=1.50] (58) to (59);
		\draw [style=directed, in=-90, out=45, looseness=1.50] (58) to (57);
		\draw [style=directed] (60) to (58);
		\draw [style=directed, bend left=15, looseness=1.00] (69) to (63);
		\draw [style=directed, bend right=15, looseness=1.00] (67) to (63);
		\draw [style=directed] (63) to (70);
		\draw [style=directed, in=-135, out=90, looseness=1.50] (65) to (64);
		\draw [style=directed, in=-45, out=90, looseness=1.50] (62) to (64);
		\draw [style=directed] (64) to (66);
	\end{pgfonlayer}
\end{tikzpicture}

%% file: figures/spider-gen.tikz
\begin{tikzpicture}
	\begin{pgfonlayer}{nodelayer}
		\node [style=white dot] (0) at (0, 0) {};
		\node [style=wire] (1) at (-0.75, -1) {};
		\node [style=wire] (2) at (0.75, -1) {};
		\node [style=wire] (3) at (-0.75, 1) {};
		\node [style=wire] (4) at (0.75, 1) {};
		\node [style=none] (5) at (0, -0.75) {...};
		\node [style=none] (6) at (0, 0.75) {...};
	\end{pgfonlayer}
	\begin{pgfonlayer}{edgelayer}
		\draw [style=directed, bend right=15, looseness=1.00] (2) to (0);
		\draw [style=directed, bend right=15, looseness=1.00] (0) to (4);
		\draw [style=directed, bend left=15, looseness=1.00] (1) to (0);
		\draw [style=directed, bend left=15, looseness=1.25] (0) to (3);
	\end{pgfonlayer}
\end{tikzpicture}

%% file: figures/spider_merge.tikz
\begin{tikzpicture}
	\begin{pgfonlayer}{nodelayer}
		\node [style=white dot] (0) at (-3.75, -0.25) {};
		\node [style=wire] (1) at (-4.5, -1.25) {};
		\node [style=wire] (2) at (-3, -1.25) {};
		\node [style=wire] (3) at (-4.25, 0.75) {};
		\node [style=wire] (4) at (-3.25, 0.75) {};
		\node [style=none] (5) at (-3.75, -1) {...};
		\node [style=none] (6) at (-3.75, 0.5) {...};
		\node [style=wire] (7) at (-2.5, 1.25) {};
		\node [style=none] (8) at (-1.75, 1) {...};
		\node [style=wire] (9) at (-1.25, -0.75) {};
		\node [style=wire] (10) at (-2.25, -0.75) {};
		\node [style=wire] (11) at (-1, 1.25) {};
		\node [style=white dot] (12) at (-1.75, 0.25) {};
		\node [style=none] (13) at (-1.75, -0.5) {...};
		\node [style=none] (14) at (0, 0) {$=$};
		\node [style=wire] (15) at (0.75, 1.25) {};
		\node [style=wire] (16) at (4.25, -1.25) {};
		\node [style=none] (17) at (3.25, 0.75) {...};
		\node [style=wire] (18) at (2, 1.25) {};
		\node [style=none] (19) at (1.75, 0.75) {...};
		\node [style=none] (20) at (3.25, -0.75) {...};
		\node [style=wire] (21) at (4.25, 1.25) {};
		\node [style=white dot] (22) at (2.5, 0) {};
		\node [style=white dot] (23) at (2.5, 0) {};
		\node [style=wire] (24) at (0.75, -1.25) {};
		\node [style=wire] (25) at (2.25, -1.25) {};
		\node [style=wire] (26) at (2.75, 1.25) {};
		\node [style=wire] (27) at (3, -1.25) {};
		\node [style=none] (28) at (2, -0.75) {...};
	\end{pgfonlayer}
	\begin{pgfonlayer}{edgelayer}
		\draw [style=directed, bend right=15, looseness=1.00] (2) to (0);
		\draw [style=directed, bend right=15, looseness=1.00] (0) to (4);
		\draw [style=directed, bend left=15, looseness=1.00] (1) to (0);
		\draw [style=directed, bend left=15, looseness=1.00] (0) to (3);
		\draw [style=directed, bend right=15, looseness=1.00] (9) to (12);
		\draw [style=directed, bend right=15, looseness=1.00] (12) to (11);
		\draw [style=directed, bend left=15, looseness=1.00] (10) to (12);
		\draw [style=directed, bend left=15, looseness=1.00] (12) to (7);
		\draw [style=directed] (0) to (12);
		\draw [style=directed, bend right=15, looseness=1.00] (25) to (23);
		\draw [style=directed, bend right=15, looseness=1.00] (23) to (18);
		\draw [style=directed, bend left=15, looseness=1.00] (24) to (23);
		\draw [style=directed, bend left=15, looseness=1.00] (23) to (15);
		\draw [style=directed, bend right=15, looseness=1.00] (16) to (22);
		\draw [style=directed, bend right=15, looseness=1.00] (22) to (21);
		\draw [style=directed, bend left=15, looseness=1.00] (27) to (22);
		\draw [style=directed, bend left=15, looseness=1.00] (22) to (26);
	\end{pgfonlayer}
\end{tikzpicture}

%% file: figures/spider_elim.tikz
\begin{tikzpicture}
	\begin{pgfonlayer}{nodelayer}
		\node [style=wire] (0) at (-1, -1) {};
		\node [style=white dot] (1) at (-1, 0) {};
		\node [style=wire] (2) at (-1, 1) {};
		\node [style=wire] (3) at (1, -1) {};
		\node [style=wire] (4) at (1, 1) {};
		\node [style=none] (5) at (0, 0) {$=$};
	\end{pgfonlayer}
	\begin{pgfonlayer}{edgelayer}
		\draw [style=directed] (0) to (1);
		\draw [style=directed] (1) to (2);
		\draw [style=directed] (3) to (4);
	\end{pgfonlayer}
\end{tikzpicture}

%% file: figures/spider_bb.tikz
\begin{tikzpicture}
	\begin{pgfonlayer}{nodelayer}
		\node [style=wire] (0) at (0, 1.25) {};
		\node [style=white dot] (1) at (0, 0) {};
		\node [style=wire] (2) at (0, -1.25) {};
		\node [style=none] (3) at (-0.5, -1.75) {};
		\node [style=none] (4) at (-0.5, 0.75) {};
		\node [style=bbox, label=B] (5) at (-0.5, -0.75) {};
		\node [style=none] (6) at (0.5, -0.75) {};
		\node [style=bbox, label=A] (7) at (-0.5, 1.75) {};
		\node [style=none] (8) at (0.5, -1.75) {};
		\node [style=none] (9) at (0.5, 0.75) {};
		\node [style=none] (10) at (0.5, 1.75) {};
	\end{pgfonlayer}
	\begin{pgfonlayer}{edgelayer}
		\draw [style=boxedge] (7) to (4.center);
		\draw [style=boxedge] (4.center) to (9.center);
		\draw [style=boxedge] (9.center) to (10.center);
		\draw [style=boxedge] (10.center) to (7);
		\draw [style=boxedge] (5) to (3.center);
		\draw [style=boxedge] (3.center) to (8.center);
		\draw [style=boxedge] (8.center) to (6.center);
		\draw [style=boxedge] (6.center) to (5);
		\draw [style=directed] (2) to (1);
		\draw [style=directed] (1) to (0);
	\end{pgfonlayer}
\end{tikzpicture}

%% file: figures/spider_merge_bb.tikz
\begin{tikzpicture}
	\begin{pgfonlayer}{nodelayer}
		\node [style=white dot] (0) at (-3.25, -0.25) {};
		\node [style=white dot] (1) at (-1.5, 0.25) {};
		\node [style=none] (2) at (0, 0) {$=$};
		\node [style=white dot] (3) at (2.25, 0) {};
		\node [style=none] (4) at (-2.75, 0.75) {};
		\node [style=none] (5) at (-2.75, 1.75) {};
		\node [style=none] (6) at (-3.75, 0.75) {};
		\node [style=bbox, label=A] (7) at (-3.75, 1.75) {};
		\node [style=bbox, label=C] (8) at (-2, 1.75) {};
		\node [style=none] (9) at (-2, 0.75) {};
		\node [style=none] (10) at (-1, 0.75) {};
		\node [style=none] (11) at (-1, 1.75) {};
		\node [style=bbox, label=B] (12) at (-3.75, -0.75) {};
		\node [style=none] (13) at (-3.75, -1.75) {};
		\node [style=none] (14) at (-2, -1.75) {};
		\node [style=bbox, label=D] (15) at (-2, -0.75) {};
		\node [style=none] (16) at (-2.75, -1.75) {};
		\node [style=none] (17) at (-1, -0.75) {};
		\node [style=none] (18) at (-2.75, -0.75) {};
		\node [style=none] (19) at (-1, -1.75) {};
		\node [style=bbox, label=B] (20) at (0.75, -0.75) {};
		\node [style=none] (21) at (3.75, -1.75) {};
		\node [style=none] (22) at (0.75, 0.75) {};
		\node [style=none] (23) at (2.75, 0.75) {};
		\node [style=bbox, label=C] (24) at (2.75, 1.75) {};
		\node [style=none] (25) at (3.75, 1.75) {};
		\node [style=none] (26) at (1.75, 1.75) {};
		\node [style=bbox, label=D] (27) at (2.75, -0.75) {};
		\node [style=none] (28) at (3.75, 0.75) {};
		\node [style=bbox, label=A] (29) at (0.75, 1.75) {};
		\node [style=none] (30) at (2.75, -1.75) {};
		\node [style=none] (31) at (3.75, -0.75) {};
		\node [style=none] (32) at (1.75, 0.75) {};
		\node [style=none] (33) at (1.75, -1.75) {};
		\node [style=none] (34) at (0.75, -1.75) {};
		\node [style=none] (35) at (1.75, -0.75) {};
		\node [style=wire] (36) at (-3.25, -1.25) {};
		\node [style=wire] (37) at (-1.5, -1.25) {};
		\node [style=wire] (38) at (-1.5, 1.25) {};
		\node [style=wire] (39) at (-3.25, 1.25) {};
		\node [style=wire] (40) at (1.25, -1.25) {};
		\node [style=wire] (41) at (3.25, -1.25) {};
		\node [style=wire] (42) at (3.25, 1.25) {};
		\node [style=wire] (43) at (1.25, 1.25) {};
	\end{pgfonlayer}
	\begin{pgfonlayer}{edgelayer}
		\draw [style=directed] (0) to (1);
		\draw [style=boxedge] (7) to (6.center);
		\draw [style=boxedge] (6.center) to (4.center);
		\draw [style=boxedge] (4.center) to (5.center);
		\draw [style=boxedge] (5.center) to (7);
		\draw [style=boxedge] (8) to (9.center);
		\draw [style=boxedge] (9.center) to (10.center);
		\draw [style=boxedge] (10.center) to (11.center);
		\draw [style=boxedge] (11.center) to (8);
		\draw [style=boxedge] (12) to (13.center);
		\draw [style=boxedge] (13.center) to (16.center);
		\draw [style=boxedge] (16.center) to (18.center);
		\draw [style=boxedge] (18.center) to (12);
		\draw [style=boxedge] (15) to (14.center);
		\draw [style=boxedge] (14.center) to (19.center);
		\draw [style=boxedge] (19.center) to (17.center);
		\draw [style=boxedge] (17.center) to (15);
		\draw [style=boxedge] (29) to (22.center);
		\draw [style=boxedge] (22.center) to (32.center);
		\draw [style=boxedge] (32.center) to (26.center);
		\draw [style=boxedge] (26.center) to (29);
		\draw [style=boxedge] (24) to (23.center);
		\draw [style=boxedge] (23.center) to (28.center);
		\draw [style=boxedge] (28.center) to (25.center);
		\draw [style=boxedge] (25.center) to (24);
		\draw [style=boxedge] (20) to (34.center);
		\draw [style=boxedge] (34.center) to (33.center);
		\draw [style=boxedge] (33.center) to (35.center);
		\draw [style=boxedge] (35.center) to (20);
		\draw [style=boxedge] (27) to (30.center);
		\draw [style=boxedge] (30.center) to (21.center);
		\draw [style=boxedge] (21.center) to (31.center);
		\draw [style=boxedge] (31.center) to (27);
		\draw [style=directed] (36) to (0);
		\draw [style=directed] (0) to (39);
		\draw [style=directed] (37) to (1);
		\draw [style=directed] (1) to (38);
		\draw [style=directed] (40) to (3);
		\draw [style=directed] (3) to (43);
		\draw [style=directed] (41) to (3);
		\draw [style=directed] (3) to (42);
	\end{pgfonlayer}
\end{tikzpicture}

%% file: figures/spider_inst1.tikz
\begin{tikzpicture}
	\begin{pgfonlayer}{nodelayer}
		\node [style=white dot] (0) at (-8.5, 0) {};
		\node [style=wire] (1) at (-6.75, 1) {};
		\node [style=white dot] (2) at (-6.75, 0) {};
		\node [style=wire] (3) at (-5.5, 1) {};
		\node [style=wire] (4) at (-4.5, 1) {};
		\node [style=white dot] (5) at (-5, 0) {};
		\node [style=none] (6) at (-7.75, -0.25) {,};
		\node [style=none] (7) at (-6, -0.25) {,};
		\node [style=none] (8) at (-3, 0) {$\cdots$};
		\node [style=none] (9) at (-4, -0.25) {,};
		\node [style=white dot] (10) at (0.5, 0) {};
		\node [style=none] (11) at (3.25, -0.25) {,};
		\node [style=none] (12) at (4.25, 0) {$\cdots$};
		\node [style=white dot] (13) at (2.25, 0) {};
		\node [style=wire] (14) at (0.5, 1) {};
		\node [style=none] (15) at (-0.5, -0.25) {,};
		\node [style=wire] (16) at (1.75, 1) {};
		\node [style=wire] (17) at (2.75, 1) {};
		\node [style=none] (18) at (1.25, -0.25) {,};
		\node [style=white dot] (19) at (-1.25, 0) {};
		\node [style=none] (20) at (-2, -0.25) {,};
		\node [style=wire] (21) at (0.5, -1) {};
		\node [style=wire] (22) at (-1.25, -1) {};
		\node [style=wire] (23) at (2.25, -1) {};
	\end{pgfonlayer}
	\begin{pgfonlayer}{edgelayer}
		\draw [style=directed] (2) to (1);
		\draw [style=directed] (5) to (3);
		\draw [style=directed] (5) to (4);
		\draw [style=directed] (10) to (14);
		\draw [style=directed] (13) to (16);
		\draw [style=directed] (13) to (17);
		\draw [style=directed] (22) to (19);
		\draw [style=directed] (21) to (10);
		\draw [style=directed] (23) to (13);
	\end{pgfonlayer}
\end{tikzpicture}

%% file: figures/phi-point.tikz
\begin{tikzpicture}
	\begin{pgfonlayer}{nodelayer}
		\node [style=none] (0) at (0, -0.75) {$\phi$};
		\node [style=none] (1) at (-2, -0.25) {};
		\node [style=none] (2) at (2, -0.25) {};
		\node [style=none] (3) at (0, -1.5) {};
		\node [style=wire, label={above:a}] (4) at (-1.5, 0.5) {};
		\node [style=wire, label={above:b}] (5) at (-0.75, 0.5) {};
		\node [style=wire, label={above:c}] (6) at (0, 0.5) {};
		\node [style=wire, label={above:d}] (7) at (0.75, 0.5) {};
		\node [style=wire, label={above:e}] (8) at (1.5, 0.5) {};
		\node [style=none] (9) at (-0.75, -0.25) {};
		\node [style=none] (10) at (1.5, -0.25) {};
		\node [style=none] (11) at (0, -0.25) {};
		\node [style=none] (12) at (-1.5, -0.25) {};
		\node [style=none] (13) at (0.75, -0.25) {};
	\end{pgfonlayer}
	\begin{pgfonlayer}{edgelayer}
		\draw (1.center) to (3.center);
		\draw (3.center) to (2.center);
		\draw (2.center) to (1.center);
		\draw [style=directed] (12.center) to (4);
		\draw [style=directed] (9.center) to (5);
		\draw [style=directed] (6.center) to (11);
		\draw [style=directed] (7) to (13.center);
		\draw [style=directed] (8) to (10.center);
	\end{pgfonlayer}
\end{tikzpicture}

%% file: figures/compl-fixed-arity.tikz
\begin{tikzpicture}
	\begin{pgfonlayer}{nodelayer}
		\node [style=arbi] (0) at (0, 0) {$\phi$};
		\node [style=wire, label={left:a}] (1) at (-0.5, 1.25) {};
		\node [style=wire, label={right:b}] (2) at (0.5, 1.25) {};
		\node [style=wire, label={right:c}] (3) at (1, -1) {};
		\node [style=wire, label={right:d}] (4) at (0, -1.5) {};
		\node [style=wire, label={left:e}] (5) at (-1, -1) {};
	\end{pgfonlayer}
	\begin{pgfonlayer}{edgelayer}
		\draw [style=directed] (0) to (1);
		\draw [style=directed] (0) to (2);
		\draw [style=directed] (3) to (0);
		\draw [style=directed] (4) to (0);
		\draw [style=directed] (5) to (0);
	\end{pgfonlayer}
\end{tikzpicture}

%% file: figures/psi-phi-contract.tikz
\begin{tikzpicture}
	\begin{pgfonlayer}{nodelayer}
		\node [style=none] (0) at (-1, -0.75) {$\psi$};
		\node [style=none] (1) at (-2.25, -0.25) {};
		\node [style=none] (2) at (0.25, -0.25) {};
		\node [style=none] (3) at (-1, -1.5) {};
		\node [style=wire, label={above:f}] (4) at (-1.75, 0.5) {};
		\node [style=none] (5) at (-1, 0) {};
		\node [style=none] (6) at (-0.25, 0) {};
		\node [style=wire, label={above:d}] (7) at (3.75, 0.5) {};
		\node [style=wire, label={above:e}] (8) at (4.5, 0.5) {};
		\node [style=none] (9) at (-1, -0.25) {};
		\node [style=none] (10) at (4.5, -0.25) {};
		\node [style=none] (11) at (-0.25, -0.25) {};
		\node [style=none] (12) at (-1.75, -0.25) {};
		\node [style=none] (13) at (3.75, -0.25) {};
		\node [style=none] (14) at (1, -0.25) {};
		\node [style=none] (15) at (3, -1.5) {};
		\node [style=none] (16) at (3, -0.75) {$\phi$};
		\node [style=none] (17) at (5, -0.25) {};
		\node [style=wire, label={above:c}] (18) at (3, 0.5) {};
		\node [style=none] (19) at (3, -0.25) {};
		\node [style=none] (20) at (1.5, 0) {};
		\node [style=none] (21) at (1.5, -0.25) {};
		\node [style=none] (22) at (2.25, 0) {};
		\node [style=none] (23) at (2.25, -0.25) {};
		\node [style=none, red, yshift=-1 pt] (24) at (0.25, 1) {a};
		\node [style=none, red] (25) at (1, 1) {b};
	\end{pgfonlayer}
	\begin{pgfonlayer}{edgelayer}
		\draw (1.center) to (3.center);
		\draw (3.center) to (2.center);
		\draw (2.center) to (1.center);
		\draw [style=directed] (12.center) to (4);
		\draw [style=directed, red] (5.center) to (9.center);
		\draw [style=directed, red] (6.center) to (11);
		\draw [style=directed] (7) to (13.center);
		\draw [style=directed] (8) to (10.center);
		\draw (14.center) to (15.center);
		\draw (15.center) to (17.center);
		\draw (17.center) to (14.center);
		\draw [style=directed] (18) to (19.center);
		\draw [style=undirected, red] (21.center) to (20.center);
		\draw [style=undirected, red] (23.center) to (22.center);
		\draw [red, style=undirected, in=90, out=90, looseness=1.00] (5.center) to (20.center);
		\draw [red, style=undirected, in=90, out=90, looseness=1.00] (22.center) to (6.center);
	\end{pgfonlayer}
\end{tikzpicture}

%% file: figures/compose-graphs.tikz
\begin{tikzpicture}
	\begin{pgfonlayer}{nodelayer}
		\node [style=arbi] (0) at (0, -0.75) {$\phi$};
		\node [style=wire, label={right:c}] (1) at (0.75, -1.5) {};
		\node [style=wire, label={[yshift=5pt]right:d}] (2) at (0, -2) {};
		\node [style=wire, label={left:e}] (3) at (-0.75, -1.5) {};
		\node [style=arbi] (4) at (0, 1) {$\psi$};
		\node [style=wire, label={[yshift=-5pt]right:f}] (5) at (0, 2) {};
	\end{pgfonlayer}
	\begin{pgfonlayer}{edgelayer}
		\draw [style=directed] (1) to (0);
		\draw [style=directed] (2) to (0);
		\draw [style=directed] (3) to (0);
		\draw [style=directed, label=a, red, in=-135, out=45, looseness=1.00] (0) to node[left, pos=0.7]{b} (4);
		\draw [style=directed, label=b, red, in=-45, out=135, looseness=1.00] (0) to node[right, pos=0.7]{a} (4);
		\draw [style=directed] (4) to (5);
	\end{pgfonlayer}
\end{tikzpicture}

%% file: figures/bb-example.tikz
\begin{tikzpicture}
	\begin{pgfonlayer}{nodelayer}
		\node [style=arbi] (0) at (0, 1) {$\psi$};
		\node [style=arbi] (1) at (0, -1) {$\phi$};
		\node [style=bbox, label=B] (2) at (-0.75, 1.75) {};
		\node [style=none] (3) at (0.75, 1.75) {};
		\node [style=none] (4) at (0.75, 0.25) {};
		\node [style=none] (5) at (-0.75, 0.25) {};
	\end{pgfonlayer}
	\begin{pgfonlayer}{edgelayer}
		\draw [style=directed, arcout={{}{0}{-60}{2mm}{1}}] (1) to (0);
		\draw [style=boxedge] (2) to (5.center);
		\draw [style=boxedge] (5.center) to (4.center);
		\draw [style=boxedge] (4.center) to (3.center);
		\draw [style=boxedge] (3.center) to (2);
	\end{pgfonlayer}
\end{tikzpicture}

%% file: figures/bb-example-cw.tikz
\begin{tikzpicture}
	\begin{pgfonlayer}{nodelayer}
		\node [style=arbi] (0) at (0, 1) {$\psi$};
		\node [style=arbi] (1) at (0, -1) {$\phi$};
		\node [style=bbox, label=B] (2) at (-0.75, 1.75) {};
		\node [style=none] (3) at (0.75, 1.75) {};
		\node [style=none] (4) at (0.75, 0.25) {};
		\node [style=none] (5) at (-0.75, 0.25) {};
	\end{pgfonlayer}
	\begin{pgfonlayer}{edgelayer}
		\draw [style=directed, arcout={{}{0}{60}{2mm}{1}}] (1) to (0);
		\draw [style=boxedge] (2) to (5.center);
		\draw [style=boxedge] (5.center) to (4.center);
		\draw [style=boxedge] (4.center) to (3.center);
		\draw [style=boxedge] (3.center) to (2);
	\end{pgfonlayer}
\end{tikzpicture}

%% file: figures/expand-order1.tikz
\begin{tikzpicture}
	\begin{pgfonlayer}{nodelayer}
		\node [style=arbi] (0) at (0, -1) {$\psi$};
		\node [style=wire, label={right:a}] (1) at (0, 1) {};
		\node [style=wire, label={right:b}] (2) at (1, 1) {};
		\node [style=bbox, label=A] (3) at (-0.5, 1.5) {};
		\node [style=none] (4) at (1.75, 0.5) {};
		\node [style=none] (5) at (-0.5, 0.5) {};
		\node [style=none] (6) at (1.75, 1.5) {};
		\node [style=wire, label={left:a'}] (7) at (-2.25, -0.5) {};
		\node [style=wire, label={left:b'}] (8) at (-2, 0.25) {};
	\end{pgfonlayer}
	\begin{pgfonlayer}{edgelayer}
		\draw [style=directed] (0) to (1);
		\draw [style=directed, arcin={{{}{-10}{-75}{2mm}{1}}}] (2) to (0);
		\draw [style=boxedge] (3) to (5.center);
		\draw [style=boxedge] (5.center) to (4.center);
		\draw [style=boxedge] (4.center) to (6.center);
		\draw [style=boxedge] (6.center) to (3);
		\draw [style=directed] (0) to (7);
		\draw [style=directed] (8) to (0);
	\end{pgfonlayer}
\end{tikzpicture}

%% file: figures/expand-order2.tikz
\begin{tikzpicture}
	\begin{pgfonlayer}{nodelayer}
		\node [style=wire, label={right:a}] (0) at (0, 1) {};
		\node [style=bbox, label=A] (1) at (-0.5, 1.5) {};
		\node [style=wire, label={right:b'}] (2) at (2.75, -0.75) {};
		\node [style=arbi] (3) at (1, -1) {$\psi$};
		\node [style=none] (4) at (-0.5, 0.5) {};
		\node [style=none] (5) at (1.75, 0.5) {};
		\node [style=wire, label={right:a'}] (6) at (2.5, 0) {};
		\node [style=wire, label={right:b}] (7) at (1, 1) {};
		\node [style=none] (8) at (1.75, 1.5) {};
	\end{pgfonlayer}
	\begin{pgfonlayer}{edgelayer}
		\draw [style=directed] (3) to (0);
		\draw [style=directed, arcin={{{}{-20}{75}{2mm}{1}}}] (7) to (3);
		\draw [style=boxedge] (1) to (4.center);
		\draw [style=boxedge] (4.center) to (5.center);
		\draw [style=boxedge] (5.center) to (8.center);
		\draw [style=boxedge] (8.center) to (1);
		\draw [style=directed] (3) to (6);
		\draw [style=directed] (2) to (3);
	\end{pgfonlayer}
\end{tikzpicture}

%% file: figures/expand-order3.tikz
\begin{tikzpicture}
	\begin{pgfonlayer}{nodelayer}
		\node [style=bbox, label=A] (0) at (-1.5, 1.5) {};
		\node [style=none] (1) at (2.75, 0.5) {};
		\node [style=wire, label={right:b}] (2) at (2, 1) {};
		\node [style=arbi] (3) at (0, -1) {$\psi$};
		\node [style=wire, label={right:a'}] (4) at (0, 0) {};
		\node [style=wire, label={right:b'}] (5) at (1.25, -1) {};
		\node [style=none] (6) at (-1.5, 0.5) {};
		\node [style=wire, label={right:a}] (7) at (-1, 1) {};
		\node [style=none] (8) at (2.75, 1.5) {};
	\end{pgfonlayer}
	\begin{pgfonlayer}{edgelayer}
		\draw [style=directed, arcout={{{}{0}{30}{5mm}{1}}}] (3) to (7);
		\draw [style=directed, arcin={{{}{0}{30}{6mm}{1}}}] (2) to (3);
		\draw [style=boxedge] (0) to (6.center);
		\draw [style=boxedge] (6.center) to (1.center);
		\draw [style=boxedge] (1.center) to (8.center);
		\draw [style=boxedge] (8.center) to (0);
		\draw [style=directed] (3) to (4);
		\draw [style=directed] (5) to (3);
	\end{pgfonlayer}
\end{tikzpicture}

%% file: figures/nested-ex.tikz
\begin{tikzpicture}
	\begin{pgfonlayer}{nodelayer}
		\node [style=arbi] (0) at (0, 2) {$\phi$};
		\node [style=arbi] (1) at (0, -0.5) {$\phi$};
		\node [style=bbox, label=B] (2) at (-0.75, 0.25) {};
		\node [style=bbox, label=A] (3) at (-1.25, 0.75) {};
		\node [style=none] (4) at (0.75, 0.25) {};
		\node [style=none] (5) at (0.75, -1.75) {};
		\node [style=none] (6) at (-0.75, -1.75) {};
		\node [style=none] (7) at (-1.25, -2) {};
		\node [style=none] (8) at (1, -2) {};
		\node [style=none] (9) at (1, 0.75) {};
		\node [style=wire, label={[yshift=-4pt]right:a}] (10) at (0, 3) {};
		\node [style=wire, label={[yshift=4pt]right:c}] (11) at (0, -1.5) {};
	\end{pgfonlayer}
	\begin{pgfonlayer}{edgelayer}
		\draw [style=boxedge] (3) to (2);
		\draw [style=boxedge] (2) to (4.center);
		\draw [style=boxedge] (4.center) to (5.center);
		\draw [style=boxedge] (5.center) to (6.center);
		\draw [style=boxedge] (6.center) to (2);
		\draw [style=boxedge] (3) to (9.center);
		\draw [style=boxedge] (9.center) to (8.center);
		\draw [style=boxedge] (8.center) to (7.center);
		\draw [style=boxedge] (7.center) to (3);
		\draw [style=directed, arcin={{B}{0}{-65}{4mm}{1}}] (1) to (0);
		\draw [draw=none, arcin={{A}{0}{75}{2mm}{1}}] (1) to (0);
		\draw [style=directed] (0) to (10);
		\draw [style=directed] (11) to (1);
	\end{pgfonlayer}
\end{tikzpicture}

%% file: figures/outinoutout.tikz
\begin{tikzpicture}
	\begin{pgfonlayer}{nodelayer}
		\node [style=wire] (0) at (0, 0.75) {};
		\node [style=arbi,scale=0.75] (1) at (0, 0) {$\phi$};
		\node [style=wire] (2) at (0.75, 0.5) {};
		\node [style=wire] (3) at (0.75, -0.5) {};
		\node [style=wire] (4) at (0, -0.75) {};
	\end{pgfonlayer}
	\begin{pgfonlayer}{edgelayer}
		\draw [style=directed] (2) to (1);
		\draw [style=directed] (1) to (0);
		\draw [style=directed] (1) to (3);
		\draw [style=directed] (1) to (4);
	\end{pgfonlayer}
\end{tikzpicture}

%% file: figures/bb-interp1.tikz
\begin{tikzpicture}
	\begin{pgfonlayer}{nodelayer}
		\node [style=arbi] (0) at (0, 1.5) {$\phi$};
		\node [style=arbi] (1) at (1.75, -1) {$\psi$};
		\node [style=bbox, label=B] (2) at (1, -0.25) {};
		\node [style=bbox, label=A] (3) at (0.5, 0.25) {};
		\node [style=none] (4) at (2.5, -0.25) {};
		\node [style=none] (5) at (2.5, -2.5) {};
		\node [style=none] (6) at (1, -2.5) {};
		\node [style=none] (7) at (0.5, -2.75) {};
		\node [style=none] (8) at (2.75, -2.75) {};
		\node [style=none] (9) at (2.75, 0.25) {};
		\node [style=none] (10) at (-1, -2.5) {};
		\node [style=arbi] (11) at (-1.75, -1) {$\psi$};
		\node [style=none] (12) at (-1, -0.25) {};
		\node [style=none] (13) at (-2.5, -2.5) {};
		\node [style=bbox, label=C] (14) at (-2.5, -0.25) {};
		\node [style=wire, label={right:a}] (15) at (0, 2.5) {};
		\node [style=wire, label={right:c}] (16) at (-1.75, -2) {};
		\node [style=wire, label={right:b}] (17) at (1.75, -2) {};
	\end{pgfonlayer}
	\begin{pgfonlayer}{edgelayer}
		\draw [style=boxedge] (3) to (2);
		\draw [style=boxedge] (2) to (4.center);
		\draw [style=boxedge] (4.center) to (5.center);
		\draw [style=boxedge] (5.center) to (6.center);
		\draw [style=boxedge] (6.center) to (2);
		\draw [style=boxedge] (3) to (9.center);
		\draw [style=boxedge] (9.center) to (8.center);
		\draw [style=boxedge] (8.center) to (7.center);
		\draw [style=boxedge] (7.center) to (3);
		\draw [style=directed, arcin={{B}{0}{-65}{4mm}{1}}] (1) to (0);
		\draw [draw=none, arcin={{A}{0}{75}{2mm}{1}}] (1) to (0);
		\draw [style=boxedge] (14) to (12.center);
		\draw [style=boxedge] (12.center) to (10.center);
		\draw [style=boxedge] (10.center) to (13.center);
		\draw [style=boxedge] (13.center) to (14);
		\draw [style=directed, arcin={{C}{0}{-30}{6mm}{1}}] (11) to (0);
		\draw [style=directed] (0) to (15);
		\draw [style=directed] (16) to (11);
		\draw [style=directed] (17) to (1);
	\end{pgfonlayer}
\end{tikzpicture}

%% file: figures/bb-interp2.tikz
\begin{tikzpicture}
	\begin{pgfonlayer}{nodelayer}
		\node [style=arbi] (0) at (0, 1.5) {$\phi$};
		\node [style=arbi] (1) at (1.75, -1) {$\psi$};
		\node [style=bbox, label=B] (2) at (1, -0.25) {};
		\node [style=bbox, label=A] (3) at (0.5, 0.25) {};
		\node [style=none] (4) at (2.5, -0.25) {};
		\node [style=none] (5) at (2.5, -2.5) {};
		\node [style=none] (6) at (1, -2.5) {};
		\node [style=none] (7) at (0.5, -2.75) {};
		\node [style=none] (8) at (2.75, -2.75) {};
		\node [style=none] (9) at (2.75, 0.25) {};
		\node [style=none] (10) at (-1, -2.5) {};
		\node [style=arbi] (11) at (-1.75, -1) {$\psi$};
		\node [style=none] (12) at (-1, -0.25) {};
		\node [style=none] (13) at (-2.5, -2.5) {};
		\node [style=bbox, label=C] (14) at (-2.5, -0.25) {};
		\node [style=wire, label={right:a}] (15) at (0, 2.5) {};
		\node [style=wire, label={right:c}] (16) at (-1.75, -2) {};
		\node [style=wire, label={right:b}] (17) at (1.75, -2) {};
		\node [style=none, gray] (18) at (-1.25, 0.25) {d};
		\node [style=none, gray] (19) at (1, 0.5) {e};
		\node [style=none] (20) at (-1.5, 0.25) {};
		\node [style=none] (21) at (-2.75, 0.75) {};
		\node [style=none] (22) at (1.25, 0.5) {};
		\node [style=none] (23) at (2.5, 1) {};
	\end{pgfonlayer}
	\begin{pgfonlayer}{edgelayer}
		\draw [style=boxedge] (3) to (2);
		\draw [style=boxedge] (2) to (4.center);
		\draw [style=boxedge] (4.center) to (5.center);
		\draw [style=boxedge] (5.center) to (6.center);
		\draw [style=boxedge] (6.center) to (2);
		\draw [style=boxedge] (3) to (9.center);
		\draw [style=boxedge] (9.center) to (8.center);
		\draw [style=boxedge] (8.center) to (7.center);
		\draw [style=boxedge] (7.center) to (3);
		\draw [style=directed, arcin={{A}{0}{75}{2mm}{1}}] (1) to (0);
		\draw [draw=none, arcin={{B}{0}{-65}{4mm}{1}}] (1) to (0);
		\draw [style=boxedge] (14) to (12.center);
		\draw [style=boxedge] (12.center) to (10.center);
		\draw [style=boxedge] (10.center) to (13.center);
		\draw [style=boxedge] (13.center) to (14);
		\draw [style=directed, arcin={{C}{0}{-30}{6mm}{1}}] (11) to (0);
		\draw [style=directed] (0) to (15);
		\draw [style=directed] (16) to (11);
		\draw [style=directed] (17) to (1);
		\draw [very thick, {gray!50!white}, ->] (21.center) to (20.center);
		\draw [very thick, {gray!50!white}, ->] (23.center) to (22.center);
	\end{pgfonlayer}
\end{tikzpicture}

%% file: figures/bb-ex1-kill.tikz
\begin{tikzpicture}
      \begin{pgfonlayer}{nodelayer}
        \node [style=arbi] (18) at (0.5, 0.5) {$\xi$};
        \node [style=arbi] (19) at (1, -1) {$\zeta$};
        \node [style=wire, label=right:e] (38) at (1,-2.25) {};
      \end{pgfonlayer}
      \begin{pgfonlayer}{edgelayer}
        \draw [style=directed] (38) to (19);
      \end{pgfonlayer}
    \end{tikzpicture}

%% file: figures/bb-ex1.tikz
\begin{tikzpicture}
      \begin{pgfonlayer}{nodelayer}
        \node [style=arbi] (2) at (2, 0.5) {$\phi$};
        \node [style=arbi] (3) at (3.25, -0.25) {$\psi$};
        \node [style=arbi] (4) at (2, 2.25) {$\xi$};
        \node [style=arbi] (5) at (2.5, -2.25) {$\zeta$};
        \node [style=wire, label=right:e] (11) at (2.5, -3.5) {};
        \node [style=bbox, label=B] (13) at (1, 1.25) {};
        \node [style=none] (14) at (4, -1) {};
        \node [style=none] (15) at (4, 1.25) {};
        \node [style=none] (16) at (1, -1) {};
      \end{pgfonlayer}
      \begin{pgfonlayer}{edgelayer}
        \draw [style=directed, arcin={{}{0}{60}{1.5mm}{1}}] (2) to (4);
        \draw [style=directed, arcin={{}{15}{-90}{1.5mm}{1}}] (2) to (5);
        \draw [style=directed] (5) to (3);
        \draw [style=directed] (3) to (2);
        \draw [style=directed] (11) to (5);
        \draw [style=boxedge] (13) to (15.center);
        \draw [style=boxedge] (15.center) to (14.center);
        \draw [style=boxedge] (14.center) to (16.center);
        \draw [style=boxedge] (16.center) to (13);
      \end{pgfonlayer}
    \end{tikzpicture}

%% file: figures/bb-ex1-exp.tikz
\begin{tikzpicture}
      \begin{pgfonlayer}{nodelayer}
        \node [style=arbi] (2) at (4.5, 0.5) {$\phi$};
        \node [style=arbi] (4) at (3, 2.25) {$\xi$};
        \node [style=arbi] (6) at (5.75, -0.25) {$\psi$};
        \node [style=arbi] (7) at (3.5, -2.25) {$\zeta$};
        \node [style=arbi] (9) at (1.75, 0.5) {$\phi$};
        \node [style=arbi] (11) at (3, -0.25) {$\psi$};
        \node [style=none] (16) at (6.25, -0.75) {};
        \node [style=bbox, label=B] (15) at (3.75, 1) {};
        \node [style=wire, label=right:e] (17) at (3.5, -3.5) {};
        \node [style=none] (19) at (6.25, 1) {};
        \node [style=none] (20) at (3.75, -0.75) {};
      \end{pgfonlayer}
      \begin{pgfonlayer}{edgelayer}
        \draw [style=directed, arcin={{}{0}{45}{2mm}{1}}] (2) to (4);
        \draw [style=directed, arcin={{}{15}{-60}{3.5mm}{1}}] (2) to (7);
        \draw [style=directed] (7) to (6);
        \draw [style=directed] (6) to (2);
        \draw [style=directed] (11) to (9);
        \draw [style=directed] (9) to (4);
        \draw [style=directed] (9) to (7);
        \draw [style=directed] (7) to (11);
        \draw [style=directed] (17) to (7);
        \draw [style=boxedge] (15) to (19.center);
        \draw [style=boxedge] (19.center) to (16.center);
        \draw [style=boxedge] (16.center) to (20.center);
        \draw [style=boxedge] (20.center) to (15);
      \end{pgfonlayer}
    \end{tikzpicture}

%% file: figures/embed-ex.tikz
\begin{tikzpicture}
	\begin{pgfonlayer}{nodelayer}
		\node [style=arbi] (0) at (-2, 2) {$\phi$};
		\node [style=none] (1) at (-1.25, -1.25) {};
		\node [style=arbi] (2) at (-2, 0.25) {$\psi$};
		\node [style=none] (3) at (-1.25, 1) {};
		\node [style=none] (4) at (-2.75, -1.25) {};
		\node [style=bbox, label=A] (5) at (-2.75, 1) {};
		\node [style=wire, label={right:a}] (6) at (-2, -0.75) {};
		\node [style=none] (7) at (0, 0) {$\hookrightarrow$};
		\node [style=none] (8) at (2.75, 1) {};
		\node [style=arbi,line width=0.4pt] (9) at (2, 0.25) {$\psi$};
		\node [style=arbi,line width=0.4pt] (10) at (2, 2) {$\phi$};
		\node [style=bbox, label=A] (11) at (1.25, 1) {};
		\node [style=none] (12) at (1.25, -2.25) {};
		\node [style=none] (13) at (2.75, -2.25) {};
		\node [style=arbi, {black!50!white}] (14) at (2, -1.5) {$\xi$};
		\node [style=none] (15) at (2, -0.75) {};
	\end{pgfonlayer}
	\begin{pgfonlayer}{edgelayer}
		\draw [style=boxedge] (5) to (3.center);
		\draw [style=boxedge] (3.center) to (1.center);
		\draw [style=boxedge] (1.center) to (4.center);
		\draw [style=boxedge] (4.center) to (5);
		\draw [style=directed, arcin={{}{0}{-45}{2mm}{1}}] (2) to (0);
		\draw [style=directed] (6) to (2);
		\draw [style=boxedge] (11) to (8.center);
		\draw [style=boxedge] (8.center) to (13.center);
		\draw [style=boxedge] (13.center) to (12.center);
		\draw [style=boxedge] (12.center) to (11);
		\draw [style=directed, arcin={{}{0}{-45}{2mm}{1}}] (9) to (10);
		\draw [style=directed] (15.center) to (9);
		\draw [style=undirected, {black!50!white}] (14) to (15.center);
	\end{pgfonlayer}
\end{tikzpicture}

%% file: figures/embed-rule-ex.tikz
\begin{tikzpicture}
	\begin{pgfonlayer}{nodelayer}
		\node [style=arbi] (0) at (-1.5, 1.5) {$\phi$};
		\node [style=none] (1) at (-0.75, -1.75) {};
		\node [style=arbi] (2) at (-1.5, -0.25) {$\psi$};
		\node [style=none] (3) at (-0.75, 0.5) {};
		\node [style=none] (4) at (-2.25, -1.75) {};
		\node [style=bbox, label=A] (5) at (-2.25, 0.5) {};
		\node [style=wire, label={right:a}] (6) at (-1.5, -1.25) {};
		\node [style=none] (7) at (0, 0) {$=$};
		\node [style=none] (8) at (1, -1.5) {};
		\node [style=arbi] (9) at (1.75, 1) {$\phi$};
		\node [style=none] (10) at (2.5, -1.5) {};
		\node [style=none] (11) at (2.5, 0) {};
		\node [style=bbox, label=A] (12) at (1, 0) {};
		\node [style=wire, label={right:a}] (13) at (1.75, -0.75) {};
	\end{pgfonlayer}
	\begin{pgfonlayer}{edgelayer}
		\draw [style=boxedge] (5) to (3.center);
		\draw [style=boxedge] (3.center) to (1.center);
		\draw [style=boxedge] (1.center) to (4.center);
		\draw [style=boxedge] (4.center) to (5);
		\draw [style=directed, arcin={{}{0}{-45}{2mm}{1}}] (2) to (0);
		\draw [style=directed] (6) to (2);
		\draw [style=boxedge] (12) to (11.center);
		\draw [style=boxedge] (11.center) to (10.center);
		\draw [style=boxedge] (10.center) to (8.center);
		\draw [style=boxedge] (8.center) to (12);
		\draw [style=directed, arcin={{}{0}{-45}{2mm}{1}}] (13) to (9);
	\end{pgfonlayer}
\end{tikzpicture}

%% file: figures/embed-big-rule-ex.tikz
\begin{tikzpicture}
	\begin{pgfonlayer}{nodelayer}
		\node [style=none] (0) at (-0.75, 0.75) {};
		\node [style=arbi] (1) at (-1.5, 0) {$\psi$};
		\node [style=arbi] (2) at (-1.5, 1.75) {$\phi$};
		\node [style=bbox, label=A] (3) at (-2.25, 0.75) {};
		\node [style=none] (4) at (-2.25, -2.5) {};
		\node [style=none] (5) at (-0.75, -2.5) {};
		\node [style=arbi, {black!50!white}] (6) at (-1.5, -1.75) {$\xi$};
		\node [style=none] (7) at (-1.5, -1) {};
		\node [style=arbi, {black!50!white}] (8) at (1.75, -1) {$\xi$};
		\node [style=arbi] (9) at (1.75, 1.5) {$\phi$};
		\node [style=none] (10) at (2.5, 0.5) {};
		\node [style=none] (11) at (1.75, 0) {};
		\node [style=bbox, label=A] (12) at (1, 0.5) {};
		\node [style=none] (13) at (1, -1.75) {};
		\node [style=none] (14) at (2.5, -1.75) {};
		\node [style=none] (15) at (0, 0) {$=$};
	\end{pgfonlayer}
	\begin{pgfonlayer}{edgelayer}
		\draw [style=boxedge] (3) to (0.center);
		\draw [style=boxedge] (0.center) to (5.center);
		\draw [style=boxedge] (5.center) to (4.center);
		\draw [style=boxedge] (4.center) to (3);
		\draw [style=directed, arcin={{}{0}{-45}{2mm}{1}}] (1) to (2);
		\draw [style=directed] (7.center) to (1);
		\draw [style=undirected, {black!50!white}] (6) to (7.center);
		\draw [style=boxedge] (12) to (10.center);
		\draw [style=boxedge] (10.center) to (14.center);
		\draw [style=boxedge] (14.center) to (13.center);
		\draw [style=boxedge] (13.center) to (12);
		\draw [style=directed, arcin={{}{0}{-45}{2mm}{1}}] (11.center) to (9);
		\draw [style=undirected, {black!50!white}] (8) to (11.center);
	\end{pgfonlayer}
\end{tikzpicture}

%% file: figures/Induc-Unit.tikz
\begin{tikzpicture}[small]
	\begin{pgfonlayer}{nodelayer}
		\node [style=comm] (0) at (0, -0.25) {};
		\node [style=wire] (1) at (0, 0.5) {};
	\end{pgfonlayer}
	\begin{pgfonlayer}{edgelayer}
		\draw [style=directed] (0) to (1);
	\end{pgfonlayer}
\end{tikzpicture}

%% file: figures/Induc-Multiply.tikz
\begin{tikzpicture}[small]
	\begin{pgfonlayer}{nodelayer}
		\node [style=wire] (0) at (0, 0.5) {};
		\node [style=wire] (1) at (0.25, -0.5) {};
		\node [style=comm] (2) at (0, 0) {};
		\node [style=wire] (3) at (-0.25, -0.5) {};
	\end{pgfonlayer}
	\begin{pgfonlayer}{edgelayer}
		\draw [style=directed] (1) to (2);
		\draw [style=directed] (3) to (2);
		\draw [style=directed] (2) to (0);
	\end{pgfonlayer}
\end{tikzpicture}

%% file: figures/Induc-SpiderBase.tikz
\begin{tikzpicture}[small]
	\begin{pgfonlayer}{nodelayer}
		\node [style=wire] (0) at (-0.75, 0.5) {};
		\node [style=arbi] (1) at (-0.75, -0.5) {};
		\node [style=written] (2) at (-0.125, 0) {:=};
		\node [style=comm] (3) at (0.5, -0.5) {};
		\node [style=wire] (4) at (0.5, 0.5) {};
	\end{pgfonlayer}
	\begin{pgfonlayer}{edgelayer}
		\draw [style=directed] (1) to (0);
		\draw [style=directed] (3) to (4);
	\end{pgfonlayer}
\end{tikzpicture}

%% file: figures/Induc-SpiderRecurs.tikz
\begin{tikzpicture}[small]
	\begin{pgfonlayer}{nodelayer}
		\node [style=wire] (0) at (-1, 1) {};
		\node [style=wire] (1) at (-0.5, -1) {};
		\node [style=arbi] (2) at (-1, 0) {};
		\node [style=wire] (3) at (-1.5, -1) {};
		\node [style=written] (4) at (0, 0) {:=};
		\node [style=comm] (5) at (1.25, 0.5) {};
		\node [style=wire] (6) at (1.25, 1.25) {};
		\node [style=wire] (7) at (2.25, -1) {};
		\node [style=bbox] (8) at (-1.75, -0.75) {};
		\node [style=none] (9) at (-1.25, -0.75) {};
		\node [style=none] (10) at (-1.25, -1.25) {};
		\node [style=none] (11) at (-1.75, -1.25) {};
		\node [style=wire] (12) at (0.75, -1) {};
		\node [style=none] (13) at (0.5, -1.25) {};
		\node [style=none] (14) at (1, -1.25) {};
		\node [style=bbox] (15) at (0.5, -0.75) {};
		\node [style=none] (16) at (1, -0.75) {};
		\node [style=arbi] (17) at (0.75, -0.25) {};
	\end{pgfonlayer}
	\begin{pgfonlayer}{edgelayer}
		\draw [style=directed, bend right=15, looseness=1.00] (1) to (2);
		\draw [style=directed, bend left=15, looseness=1.00] (3) to (2);
		\draw [style=directed] (2) to (0);
		\draw [style=directed, bend right=15, looseness=1.00] (7) to (5);
		\draw [style=directed] (5) to (6);
		\draw [style=boxedge] (8) to (9.center);
		\draw [style=boxedge] (9.center) to (10.center);
		\draw [style=boxedge] (10.center) to (11.center);
		\draw [style=boxedge] (11.center) to (8);
		\draw [style=directed] (12) to (17);
		\draw [style=boxedge] (15) to (16.center);
		\draw [style=boxedge] (16.center) to (14.center);
		\draw [style=boxedge] (14.center) to (13.center);
		\draw [style=boxedge] (13.center) to (15);
		\draw [style=directed] (17) to (5);
	\end{pgfonlayer}
\end{tikzpicture}

%% file: figures/Induc-MergeBase-eq.tikz
\begin{tikzpicture}[small]
	\begin{pgfonlayer}{nodelayer}
		\node [style=wire] (0) at (1, -0.75) {};
		\node [style=wire] (1) at (1.5, 1.25) {};
		\node [style=none] (2) at (1.25, -1) {};
		\node [style=none] (3) at (1.25, -0.5) {};
		\node [style=bbox, label=A] (4) at (0.75, -0.5) {};
		\node [style=arbi] (5) at (1.5, 0.25) {};
		\node [style=none] (6) at (0.75, -1) {};
		\node [style=arbi] (7) at (-1.5, 0.25) {};
		\node [style=wire] (8) at (-2, -0.75) {};
		\node [style=wire] (9) at (-1.5, 1.25) {};
		\node [style=arbi] (10) at (-0.75, -0.25) {};
		\node [style=none] (11) at (-1.75, -0.5) {};
		\node [style=bbox, label=A] (12) at (-2.25, -0.5) {};
		\node [style=none] (13) at (-1.75, -1) {};
		\node [style=none] (14) at (-2.25, -1) {};
		\node [style=written] (15) at (0, 0) {=};
	\end{pgfonlayer}
	\begin{pgfonlayer}{edgelayer}
		\draw [style=directed] (5) to (1);
		\draw [style=boxedge] (4) to (3.center);
		\draw [style=boxedge] (3.center) to (2.center);
		\draw [style=boxedge] (2.center) to (6.center);
		\draw [style=boxedge] (6.center) to (4);
		\draw [style=directed, arcin={{}{0}{-45}{2mm}{1}}] (0) to (5);
		\draw [style=directed] (10) to (7);
		\draw [style=directed] (7) to (9);
		\draw [style=boxedge] (12) to (11.center);
		\draw [style=boxedge] (11.center) to (13.center);
		\draw [style=boxedge] (13.center) to (14.center);
		\draw [style=boxedge] (14.center) to (12);
		\draw [style=directed, arcin={{}{0}{-45}{2mm}{1}}] (8) to (7);
	\end{pgfonlayer}
\end{tikzpicture}

%% file: figures/Induc-MergeLemma.tikz
\begin{tikzpicture}[small]
	\begin{pgfonlayer}{nodelayer}
		\node [style=bbox, label=B] (0) at (2.25, -0.75) {};
		\node [style=wire] (1) at (1, -1) {};
		\node [style=none] (2) at (2.25, -1.25) {};
		\node [style=wire] (3) at (1.75, 1.25) {};
		\node [style=none] (4) at (1.25, -1.25) {};
		\node [style=none] (5) at (2.75, -1.25) {};
		\node [style=none] (6) at (2.75, -0.75) {};
		\node [style=none] (7) at (1.25, -0.75) {};
		\node [style=bbox, label=A] (8) at (0.75, -0.75) {};
		\node [style=arbi] (9) at (1.75, 0.25) {};
		\node [style=none] (10) at (0.75, -1.25) {};
		\node [style=wire] (11) at (2.5, -1) {};
		\node [style=none] (12) at (-0.75, -1) {};
		\node [style=arbi] (13) at (-1.75, 0.25) {};
		\node [style=wire] (14) at (-2.25, -0.75) {};
		\node [style=wire] (15) at (-1.75, 1.25) {};
		\node [style=arbi] (16) at (-1, -0.25) {};
		\node [style=none] (17) at (-2, -0.5) {};
		\node [style=bbox, label=B] (18) at (-1.25, -1) {};
		\node [style=bbox, label=A] (19) at (-2.5, -0.5) {};
		\node [style=none] (20) at (-1.25, -1.5) {};
		\node [style=none] (21) at (-2, -1) {};
		\node [style=none] (22) at (-2.5, -1) {};
		\node [style=none] (23) at (-0.75, -1.5) {};
		\node [style=wire] (24) at (-1, -1.25) {};
		\node [style=written] (25) at (0, 0) {=};
	\end{pgfonlayer}
	\begin{pgfonlayer}{edgelayer}
		\draw [style=directed] (9) to (3);
		\draw [style=boxedge] (8) to (7.center);
		\draw [style=boxedge] (7.center) to (4.center);
		\draw [style=boxedge] (4.center) to (10.center);
		\draw [style=boxedge] (10.center) to (8);
		\draw [style=boxedge] (0) to (6.center);
		\draw [style=boxedge] (6.center) to (5.center);
		\draw [style=boxedge] (5.center) to (2.center);
		\draw [style=boxedge] (2.center) to (0);
		\draw [style=directed, arcin={{}{0}{-30}{3mm}{1}}] (1) to (9);
		\draw [style=directed, arcin={{}{0}{-30}{3mm}{1}}] (11) to (9);
		\draw [style=directed] (16) to (13);
		\draw [style=directed] (13) to (15);
		\draw [style=boxedge] (19) to (17.center);
		\draw [style=boxedge] (17.center) to (21.center);
		\draw [style=boxedge] (21.center) to (22.center);
		\draw [style=boxedge] (22.center) to (19);
		\draw [style=boxedge] (18) to (12.center);
		\draw [style=boxedge] (12.center) to (23.center);
		\draw [style=boxedge] (23.center) to (20.center);
		\draw [style=boxedge] (20.center) to (18);
		\draw [style=directed, arcin={{}{0}{-45}{2mm}{1}}] (14) to (13);
		\draw [style=directed, arcin={{}{0}{-45}{2mm}{1}}] (24) to (16);
	\end{pgfonlayer}
\end{tikzpicture}

%% file: figures/Induc-MergeStep-eq.tikz
\begin{tikzpicture}[small]
	\begin{pgfonlayer}{nodelayer}
		\node [style=bbox, label=B] (0) at (2, -0.5) {};
		\node [style=wire] (1) at (1.25, -0.25) {};
		\node [style=none] (2) at (2, -1) {};
		\node [style=wire] (3) at (2.25, 1.75) {};
		\node [style=none] (4) at (1.5, -0.5) {};
		\node [style=none] (5) at (2.5, -1) {};
		\node [style=none] (6) at (2.5, -0.5) {};
		\node [style=none] (7) at (1.5, 0) {};
		\node [style=bbox, label=A] (8) at (1, 0) {};
		\node [style=arbi] (9) at (2.25, 0.75) {};
		\node [style=none] (10) at (1, -0.5) {};
		\node [style=wire] (11) at (2.25, -0.75) {};
		\node [style=none] (12) at (-1.5, -0.5) {};
		\node [style=arbi] (13) at (-1.75, 0.75) {};
		\node [style=wire] (14) at (-2.75, -0.25) {};
		\node [style=wire] (15) at (-1.75, 1.75) {};
		\node [style=arbi] (16) at (-1, 0) {};
		\node [style=none] (17) at (-2.5, 0) {};
		\node [style=bbox, label=B] (18) at (-2, -0.5) {};
		\node [style=bbox, label=A] (19) at (-3, 0) {};
		\node [style=none] (20) at (-2, -1) {};
		\node [style=none] (21) at (-2.5, -0.5) {};
		\node [style=none] (22) at (-3, -0.5) {};
		\node [style=none] (23) at (-1.5, -1) {};
		\node [style=wire] (24) at (-1.75, -0.75) {};
		\node [style=wire] (25) at (-0.25, -0.75) {};
		\node [style=written] (26) at (0, 0) {=};
		\node [style=wire] (27) at (3.25, -0.25) {};
	\end{pgfonlayer}
	\begin{pgfonlayer}{edgelayer}
		\draw [style=directed] (9) to (3);
		\draw [style=boxedge] (8) to (7.center);
		\draw [style=boxedge] (7.center) to (4.center);
		\draw [style=boxedge] (4.center) to (10.center);
		\draw [style=boxedge] (10.center) to (8);
		\draw [style=boxedge] (0) to (6.center);
		\draw [style=boxedge] (6.center) to (5.center);
		\draw [style=boxedge] (5.center) to (2.center);
		\draw [style=boxedge] (2.center) to (0);
		\draw [style=directed, arcin={{}{0}{-30}{3mm}{1}}] (1) to (9);
		\draw [style=directed, arcin={{}{0}{-25}{4mm}{1}}] (11) to (9);
		\draw [style=directed] (16) to (13);
		\draw [style=directed] (13) to (15);
		\draw [style=boxedge] (19) to (17.center);
		\draw [style=boxedge] (17.center) to (21.center);
		\draw [style=boxedge] (21.center) to (22.center);
		\draw [style=boxedge] (22.center) to (19);
		\draw [style=boxedge] (18) to (12.center);
		\draw [style=boxedge] (12.center) to (23.center);
		\draw [style=boxedge] (23.center) to (20.center);
		\draw [style=boxedge] (20.center) to (18);
		\draw [style=directed, arcin={{}{0}{-45}{2mm}{1}}] (14) to (13);
		\draw [style=directed, arcin={{}{0}{-45}{2mm}{1}}] (24) to (16);
		\draw [style=directed] (25) to (16);
		\draw [style=directed] (27) to (9);
	\end{pgfonlayer}
\end{tikzpicture}